\documentclass{article}
\usepackage{graphicx}
\usepackage{amsfonts}
\usepackage{amsmath}
\usepackage{amssymb}
\usepackage{url}
\usepackage{wrapfig}
\usepackage{float}
\usepackage{fancyhdr}
\usepackage{indentfirst}
\usepackage{enumerate}
\usepackage{epstopdf}
\usepackage{dsfont}
\usepackage[colorlinks=true,citecolor=blue]{hyperref}
\usepackage{amsthm}
\usepackage{multirow}
\usepackage{color}
\usepackage{natbib}
\usepackage{comment}
\usepackage{caption}

\def\d{\mathrm{d}}
\def\laweq{\buildrel \mathrm{d} \over =}

\DeclareMathOperator*{\argmax}{arg\,max}
\DeclareMathOperator*{\argmin}{arg\,min}
\newcommand{\X}{\mathcal {X}}

\newcommand{\var}{\mathrm{var}}

\newcommand{\VaR}{\mathrm{VaR}}

\newcommand{\ES}{\mathrm{ES}}

\newcommand{\E}{\mathbb{E}}

\newcommand{\R}{\mathbb{R}}
\newcommand{\N}{\mathbb{N}}
\newcommand{\p}{\mathbb{P}}

\newcommand{\id}{\mathds{1}}

\renewcommand{\ge}{\geqslant}
\renewcommand{\le}{\leqslant}
\renewcommand{\geq}{\geqslant}
\renewcommand{\leq}{\leqslant}
\renewcommand{\epsilon}{\varepsilon}
\newcommand{\esssup}{\mathrm{ess\mbox{-}sup}}
\newcommand{\essinf}{\mathrm{ess\mbox{-}inf}}

\theoremstyle{plain}
\newtheorem{theorem}{Theorem}

\newtheorem{proposition}{Proposition}
\theoremstyle{definition}
\newtheorem{definition}{Definition}
\newtheorem{example}{Example}
\usepackage{tkz-graph}
\usetikzlibrary{patterns, positioning, arrows}
\usetikzlibrary{decorations.markings}

\theoremstyle{remark}
\newtheorem{remark}{Remark}

\newcommand{\cet}{\begin{center}}
	\newcommand{\ecet}{\end{center}}
\begin{document}

	\captionsetup[figure]{labelfont={ },name={Figure},labelsep=period}
	\captionsetup[table]{labelfont={},name={Table},labelsep=period}
	\title{Diversification quotients  based on VaR and ES}
	
	\author{ Xia Han \thanks{School of Mathematical Sciences and LPMC, Nankai University, China.
			\texttt{xiahan@nankai.edu.cn}} 
		\and Liyuan Lin\thanks{Department of Statistics and Actuarial Science, University of Waterloo, Canada.   \texttt{l89lin@uwaterloo.ca}}
		\and Ruodu Wang\thanks{Department of Statistics and Actuarial Science, University of Waterloo, Canada.   \texttt{wang@uwaterloo.ca}}}
	\date{\today}
	\maketitle
	
	\begin{abstract}
 The diversification  quotient (DQ) is recently introduced for  quantifying  the degree of diversification of a stochastic portfolio model.  It has an axiomatic foundation and can be defined through  a parametric class of risk measures.   Since the Value-at-Risk (VaR) and the Expected Shortfall (ES) are the most prominent risk measures widely used in both banking and insurance,   we investigate  DQ constructed from VaR and ES in this paper.   In particular,  for the popular models of  elliptical and multivariate regular varying  (MRV) distributions, explicit formulas are available. 
The   portfolio optimization  problems for the elliptical  and MRV models are  also  studied.
  Our results further
  reveal favourable features of DQ, both theoretically and practically, compared to traditional diversification indices based on a single risk measure.
		
	\end{abstract}

	\section{Introduction}
 In order to mitigate risks in portfolios of financial investment quantitatively, a common approach is to compute a quantitative index of the portfolio model, based on e.g., the volatility, variance, an expected utility or a risk measure, following the seminal idea of \cite{M52} on portfolio diversification.  
 In the literature, one of the most prominent examples of the diversification index based on   a general risk measure   is defined by \cite{T07}  which is referred as diversification ratio (DR). \cite{CC08} investigated the theoretical and empirical properties of DR  in portfolio construction and compared the behavior of the resulting portfolio to  common, simple strategies. 
 See \cite{EWW15} 
and  \cite{KD22} for theories of DR and other diversification indices. 
 \cite{BDI08} defined a closely related notion of DR which is called the diversification gain and explored various methods of modeling dependence and their
influence on diversification gain.

Different from the traditional diversification indices such as DR  in the above literature,  \cite{HLW22} proposed six axioms -- non-negativity, location
invariance, scale invariance, rationality, normalization and continuity -- which jointly characterize   a new diversification index, called the diversification quotient (DQ),  whose definition is based on a class of risk measures decreasing in an index $\alpha$. 
All commonly used risk measures belong to a monotonic parametric family, and this includes VaR, ES, expectiles, mean-variance, and entropic risk measures.  They argued that DQ has many appealing features both theoretically and practically,  while these properties, in particular the six axioms above, are not shared by DR based  on VaR, ES, or any other commonly used risk measure. {Moreover,  portfolio optimization of DQs based on VaR and ES can be
computed very efficiently, and thus DQ can be easily applied to real data.} 

Most properties of DQ  are studied by \cite{HLW22} for a general class of  risk  measures. In this paper, we focus on specific risk measures, in particular, the
 Value-at-Risk (VaR) and the Expected Shortfall  (ES). {Even though VaR has been  criticized because of its lack of subadditivity and ES   requires the loss to have a finite mean,   VaR and ES  are still the two most common classes of risk measures in practice, widely employed in global banking and insurance regulatory frameworks; see Basel III/IV (\cite{BASEL19}) and Solvency II (\cite{E11}).   More  theoretical properties and discussions  of    VaR and  ES    can be found in,  e.g., \cite{ADEH99}, \cite{EPRWB14,ELW18}, \cite{EKT15} and the references therein.}   
We pay particular attention to two popular models in finance and insurance, namely, 
  elliptical and  multivariate regular variation (MRV)  distributions. 
  Elliptical distributions, including normal and t-distributions as special cases, are the most standard tools for quantitative risk management (\cite{MFE15}). 
  They have been studied for DR with convenient properties; see \cite{CTYZ22} and the references therein. 
  The MRV model is widely used in 
  Extreme Value Theory for investigating the portfolio diversification; see, e.g., \cite{MR10}, \cite{ME13} and \cite{BMWW16}. 
  
 This paper is an extension of  \cite{HLW22} in which  an axiomatic framework  of diversification indices is proposed and general properties of DQ are  studied. As a new concept of diversification index,  studying  properties  such as explicit formulas and  limiting behavior of DQ under    specific risk measures and special risk models  will help us to better understand and use DQ in risk management applications. In addition, the advantages of DQ and the connection between DQ and DR are clearer under the elliptical and MRV  models, revealing many attractive features of choosing DQ instead of DR to quantify diversification risk, especially for tail heaviness and common shocks.

The paper is organized as follows. In Section \ref{sec:2}, the definition of DQ  and some preliminaries on risk measures are collected. In Section \ref{VaR-ES}, we study general properties for  DQs based on VaR and ES.   Since DQs  based on VaR and ES have  natural ranges 
of  $[0, n]$  and  $[0, 1]$, respectively, some special dependence  structures of the portfolio that correspond to the special values of $0$, $1$, and $n$ are constructed with clear interpretation for values in between (Theorem \ref{th:var-01n}). In Section \ref{large-port}, we focus on  DQ for large portfolios. By the Law of Large Numbers, we show that DQs based VaR and ES  for a portfolio with independent components tend to 0 as the number of assets in the portfolio increases to infinity (Theorem \ref{pro:ind}). The limits for DQs based on VaR and ES for   portfolios with exchangeable components  do not necessarily tend to 0. We show that the upper bound for the limit decreases in  the bivariate correlation coefficient. (Proposition \ref{prop:exchangeable}).  
In Section \ref{sec:4},  DQ is applied to  elliptical models;  explicit formulas  and the limiting behavior of DQs based on VaR and ES are available (Proposition \ref{prop:comp_Dvar} and Theorem \ref{thm:MRV}). 
 Moreover,  we present several  numerical results for  the two most important elliptical distributions used in finance and insurance, namely the multivariate normal distribution and the multivariate t-distribution, and show that DQ can properly capture tail heaviness. 
 As a popular tool for modeling heavy-tailed phenomena,  MRV models  for DQ are studied in Section \ref{sec:MRV}. 
 Furthermore, we generalize the results to the optimal portfolio selection problem  in Section \ref{sec:5}. Under elliptical models, the optimization problem can boil down to a well-studied problem (see e.g.,  \cite{CC08}) and a limiting result in MRV models is also derived (Theorem \ref{th:opt_elli} and Proposition \ref{prop:lim}). 
 We conclude the paper in Section  \ref{sec:6}.

	\section{Diversification quotients}\label{sec:2}
	
 
	 Throughout this paper,  $(\Omega,\mathcal F,\p)$ is an atomless probability. The atomless assumption in our context  is very weak   and  it is widely used in statistics and risk management; see  \cite{D02} and  Section A.3 of \cite{FS16}  for details of atomless probability spaces.  Almost surely equal random variables are treated as identical.  A risk measure  $\phi$ is a mapping  from $\X$ to $\R$, where $\X$ is  a convex cone of random variables  on  $(\Omega,\mathcal F,\p)$ representing losses faced by a financial institution or an investor, and $\X$ is assumed to include all constants (i.e., degenerate random variables). For $p\in (0,\infty)$, denote by $L^p=L^p(\Omega,\mathcal F,\p)$ the set of all random variables $X$ with $\E[|X|^p]<\infty$ where $\E$ is the expectation under $\p$. Furthermore, $L^\infty=L^\infty(\Omega,\mathcal F,\p)$ is the space of all essentially bounded random variables, and $L^0=L^0(\Omega,\mathcal F,\p)$ is the space of all random variables.  Write $X\sim F$  if the random variable $X$ has the distribution function $F$  under $\mathbb{P}$, and 
	$X \laweq  Y$ if two random variables $X$ and $Y$ have the same distribution. 
	We always write $\mathbf X=(X_1,\dots,X_n)$ and  $\mathbf 0 $ for the $n$-vector of zeros. 
	Further, denote by $[n]=\{1,\dots,n\}$, $\R_+=[0,\infty)$  and $\overline\R=[-\infty,\infty]$.  Terms such as increasing or decreasing functions are in the non-strict sense. For $X\in \X$, $\esssup (X)$ and $\essinf(X)$ are  the essential supremum and the essential infimum of $X$, respectively. 
	
	 A \emph{diversification index} $D$ is a mapping from  $\X^n$ to $\overline\R$, which is used to  quantify the magnitude of diversification of a risk vector $\mathbf X\in\X^n$ representing portfolio losses.
	Our convention is that  a smaller value of $D(\mathbf X)$   represents a stronger diversification. 
	Measuring diversification is closely related to risk measures.  Some standard properties of a risk measure $\phi:\X\to \R$ are collected below.
\begin{itemize}
\item[{[M]}] Monotonicity: $\phi(X)\le \phi(Y)$ for all $X,Y \in \mathcal X$ with $X\le Y$.
\item[{$[\mathrm{CA}]$}] Constant additivity:  $\phi(X+c)=\phi(X)+ c$ for all  $c\in \R$ and $X\in \X$.
\item[{$[\mathrm{PH}]$}] Positive homogeneity:  $\phi(\lambda X)=\lambda \phi(X)$ for all $\lambda \in (0, \infty)$ and $X\in \X$.
\item[{[SA]}] Subadditivity:  $\phi(X+Y)\le \phi(X)+\phi(Y)$ for all $X,Y\in \X$.
\end{itemize}

The two popular classes of risk measures in banking and insurance practice are VaR and ES.  The VaR at level $\alpha \in [0,1)$ is defined as$$
	\VaR_\alpha(X)=\inf\{x\in \R: \p(X\le x) \ge 1-\alpha\},~~~X\in L^0,
	$$
	and the ES (also called CVaR, TVaR or AVaR) at level $\alpha \in (0,1)$ is defined as
	$$
	\ES_{\alpha}(X) = \frac 1 \alpha \int_{0}^\alpha \VaR_\beta(X) \d \beta,~~~X\in L^1,
	$$
	and  $\ES_0(X)=\esssup(X)=\VaR_0(X)$ which may be $\infty$.
	The probability level $\alpha$ above is typically very small, e.g., $0.01$ or $0.025$ in \cite{BASEL19}; note that we use the  ``small $\alpha$"  convention as in \cite{HLW22}.  Both VaR and ES  satisfy the  properties [M], [CA] and [PH], while ES also satisfies the property [SA].

 To measure diversification {quantitatively},  a new index, called diversification quotient (DQ), is introduced as follows.

	\begin{definition}[\cite{HLW22}] Let $\rho=\left(\rho_{\alpha}\right)_{\alpha \in I}$ be a class of risk measures indexed by $\alpha \in I=$ $(0, \bar{\alpha})$ with $\bar{\alpha} \in(0, \infty]$ such that $\rho_{\alpha}$ is decreasing in $\alpha$. For $\mathbf{X} \in \mathcal{X}^{n}$, the diversification quotient based on the class $\rho$ at level $\alpha \in I$ is defined by
		$$
		\operatorname{DQ}_{\alpha}^{\rho}(\mathbf{X})=\frac{\alpha^{*}}{\alpha}, \quad \text { where } \alpha^{*}=\inf \left\{\beta \in I: \rho_{\beta}\left(\sum_{i=1}^{n} X_{i}\right) \leq \sum_{i=1}^{n} \rho_{\alpha}\left(X_{i}\right)\right\}
		$$
		with the convention $\inf (\varnothing)=\bar{\alpha}.$\label{def:1}
	\end{definition}

 \begin{remark}
The value of $\mathrm{DQ}^\rho_\alpha$ depends on how the class $\rho=(\rho_\alpha)_{\alpha \in I}$ is parametrized. For instance, one could, hypothetically, use a different parametrization $\VaR'_\alpha=\VaR_{\alpha^2}$ for the class VaR, although there is no real reason to do so. The value of $\mathrm{DQ}^{\VaR'}_\alpha$ is generally different from $\mathrm{DQ}^\VaR_{\alpha^2}$, but they generate the same order; that is,
 $\mathrm{DQ}^{\VaR'}_\alpha(\mathbf X) \le \mathrm{DQ}^{\VaR'}_\alpha(\mathbf Y) $
 if and only if
  $\mathrm{DQ}^\VaR_{\alpha^2} (\mathbf{X}) \le \mathrm{DQ}^\VaR_{\alpha^2} (\mathbf {Y})$, which can be checked by definition.
Therefore, different parametrizations do not affect the application of DQ in portfolio optimization.
\end{remark}

\citet[Theorem 1]{HLW22} characterized a subclass of DQ via six axioms: non-negativity, location
invariance, scale invariance, rationality, normalization and continuity; such DQs are defined on the class of risk measures satisfying [M], [CA] and [PH]. 
DQ is defined based on a monotonic parametric class of risk measures. All commonly used risk measures belong to a monotonic parametric family; for instance, this includes VaR, ES,  expectiles, mean-variance, and entropic risk measures; see \cite{FS16} for a general treatment of risk measures.


 In finance and insurance, the risk measures VaR and ES play  prominent roles, as they are  specified in regulatory documents such as \cite{BASEL19} and \cite{E11}. We will focus on VaR or ES as the risk measures  assessing diversification by DQ in this paper.  In particular, both VaR and ES  satisfy the  properties [M], [CA] and [PH],  and hence 
	 $\mathrm{DQ}^{\VaR}_\alpha$ and  $\mathrm{DQ}^{\ES}_\alpha$ satisfy the six above axioms.
  
Another popular diversification index is the diversification ratio (see e.g., \cite{T07} and \cite{EWW15}), defined as  \begin{equation}\label{eq:DR}
		{\rm DR}^{\phi}(\mathbf X)= \frac{ \phi \left(\sum_{i=1}^n X_i\right)}{ \sum_{i=1}^n \phi(X_i)},
	\end{equation}
	where $\phi$ is a suitably chosen risk measure, such as   $\VaR_\alpha$, $\ES_\alpha$, var  or SD. 
	Although DR generally does not satisfy some of the six axioms, we will compare DQ and DR in several parts of the paper.

	\section{DQ based on VaR and ES}\label{VaR-ES}
	
	In this section,   we will focus on the theoretical properties of ${\rm DQ}^{\VaR}_{\alpha}$ and ${\rm DQ}^{\ES}_{\alpha}$. For VaR and ES, the  interval   in Definition \ref{def:1} has a natural range of $I=(0,1)$. Similarly to \cite{HLW22}, we let $\X^n$ be $(L^0)^n$ when we discuss ${\rm DQ}^{\VaR}_{\alpha}$ and $(L^1)^n$ when we discuss ${\rm DQ}^{\ES}_{\alpha}$. 
	To compute ${\rm DQ}^{\ES}_\alpha$, we first define the {superquantile transform} (\citet[Example 4]{LSW21}).
 The term ``superquantile" is an alternative name for ES; see \cite{RRM14}.  
	
	\begin{definition}
	    	The \emph{superquantile transform}   of a distribution $F$ with finite mean
	is  a distribution $\widetilde F$ with quantile function $p\mapsto \ES_{1-p}(X)$ for $p\in(0,1)$, where $X\sim F$. 
	\end{definition}
The following alternative formulas for DQs based on VaR and ES  will be useful later. They are shown in Theorem 3 of \cite{HLW22}.
  For a given $\alpha \in (0,1)$, ${\rm DQ}^{\VaR}_\alpha $ and ${\rm DQ}^{\ES}_\alpha$ can be computed by
		\begin{equation} {\rm DQ}^{\VaR}_\alpha (\mathbf X)=\frac{1-   F \left(\sum_{i=1}^n \VaR_{\alpha}(X_i)\right) }{\alpha}
			\mbox{~~and~~}
			{\rm DQ}^{\ES}_\alpha (\mathbf X)=\frac{1- \widetilde F \left(\sum_{i=1}^n \ES_{\alpha}(X_i)\right) }{\alpha} , \label{eq:superquantile}
		\end{equation}
		where $F$ is the distribution of $\sum_{i=1}^n X_i$ and $\widetilde F$ is the {superquantile transform} of $F$.  

\begin{remark}\label{rem:superquantile}
 Let  $S =\sum_{i=1}^n X_i$. 
If $S$ has a continuous and strictly monotone quantile function, then
\eqref{eq:superquantile} can be rewritten  as
$$
{\rm DQ}^{\VaR}_\alpha (\mathbf X)  =\frac{1}{\alpha} \p\left(S> \sum_{i=1}^n \VaR_{\alpha}(X_i)\right), ~~~\mathbf X\in \X^n,$$
and
$$
{\rm DQ}^{\ES}_\alpha (\mathbf X) =\frac{1}{\alpha} \mathbb Q\left(S> \sum_{i=1}^n \ES_{\alpha}(X_i)\right),~~~\mathbf X\in \X^n, 
$$
for some probability measure $\mathbb Q$. 
To give a formula for $\mathbb Q$, let $F$ be the distribution of $S$, and $\alpha_0=1-F(\E[S]).$
There exists an increasing and continuous function $g:(0,1)\to [0,1]$
such that $\ES_{g(\alpha)}(S)=\VaR_\alpha(S)\mbox{~for all }\alpha \in(0,\alpha_0)$ and $g(\alpha)=1$ for $\alpha \in [\alpha_0,1)$.
We can express  $\mathbb Q$  by 
$
 {\d \mathbb Q}/{\d \p} = g'(1-F(S)).
$ 
\end{remark}

\begin{remark}
DQ based on $\ES$ admits another convenient formula in \citet[Theorem 3]{HLW22}. 
  If $\p(\sum_{i=1}^n  X_i>\sum_{i=1}^n \ES_\alpha(X_i) )>0$, then 
 \begin{equation}\label{eq:op_ES}{\rm DQ}^{\ES}_\alpha(\mathbf X)= \frac{1}{\alpha}\min_{r\in (0,\infty)} \E\left[\left(r \sum_{i=1}^n (X_i-\ES_\alpha(X_i))+1\right)_+\right],
 \end{equation}
and otherwise ${\rm DQ}^{\ES}_\alpha(\mathbf X)=0.$
 The main advantage of this formula of $\mathrm{DQ}^\ES_\alpha$ is computation and optimization. In particular, this formula allows us to write the portfolio optimization problem of $\mathrm{DQ}^\ES_\alpha$ as a convex program; this is shown in Proposition 5 of \cite{HLW22}.
\end{remark}

Next, we see that if $\alpha\in(0,1/n)$,  there are three special values of ${\rm DQ}^{\VaR}_\alpha$,
	which are $0$, $1$ and $n$, corresponding to different representative dependence structures.
	The last value of $n$
	is based on a useful inequality
	\begin{equation}\label{eq:e1p} \VaR_{n\alpha} \left(\sum_{i=1}^n X_i\right) \le \sum_{i=1}^n \VaR_{\alpha}(X_i) \end{equation}
	from Corollary 1 of \cite{ELW18},
	and its sharpness is stated in Corollary 2 therein. 
	For ${\rm DQ}^{\ES}_\alpha$, there are two special numbers, $0$ and $1$,  because ES is a class of  subadditive risk measures.
	As a natural question, we wonder for what types of dependence structures these special values are attained. Next, we address this question.
	
	We first present the concept of risk concentration in \cite{WZ20} which will be useful to understand the dependence structures corresponding to special  values of ${\rm DQ}^{\VaR}_\alpha$ and ${\rm DQ}^{\ES}_\alpha$.

	\begin{definition}[Tail event and $\alpha$-concentrated]
		Let $X$ be a random variable and $\alpha \in(0,1)$.
		\begin{itemize}
			\item[(i)] A tail event of  $X$ is an event $A \in \mathcal{F}$ with $0<\p(A)<1$ such that $X(\omega)\ge X(\omega')$ holds for a.s.~all  $\omega \in A$ and $\omega' \in A^c$, where $A^c$ stands for the complement of $A$.
			\item[(ii)] A random vector $(X_1,\dots,X_n)$ is $\alpha$-concentrated if its component share a common tail event of probability $\alpha$.\footnote{\cite{WZ20} used the ``large $\alpha$" convention, and hence our $\alpha$-concentration corresponds to their $(1-\alpha)$-concentration.}
		\end{itemize}
	\end{definition}

	Theorem 4 of \cite{WZ20} gives that a random vector  $(X_1,\dots, X_n)$  is $\alpha$-concentrated for all $\alpha \in (0,1)$ if and only if it is comonotonic, and hence the dependence notion  of $\alpha$-concentration is weaker than comonotonicity. 
A random vector $(X_1,\dots,X_n)$ is \emph{comonotonic}
if there exists a random variable $Z$ and increasing functions $f_1,\dots,f_n$ on $\R$ such that $X_i=f_i(Z)$ a.s.~for every $i\in [n]$.

 	We first address the case that ${\rm DQ}^{\VaR}_{\alpha}(\mathbf X)=n$, which involves the dependence concepts of both risk concentration and mutual exclusivity (see \cite{DD99}).
	Thus, to arrive at the maximum value of ${\rm DQ}^{\VaR}_{\alpha}(\mathbf X)=n$, one requires a dependence structure that is a combination of positive and negative dependence. This phenomenon is common in problems in VaR aggregation; see \cite{PW15} for extremal dependence concepts. 
	For this purpose, we propose 
	the \emph{$\alpha$-concentration-exclusion} ($\alpha$-CE)  model for $\alpha \in (0,1/n)$, which is a random vector $\mathbf X\in \X^n$ satisfying   four conditions: 

\begin{enumerate}[(i)]
    \item 	$\p\left(X_i>\VaR_{\alpha}(X_i)\right)=\alpha$;
    \item 
	$\p(X_i\ge \VaR_{\alpha}(X_i))\ge n\alpha$;
	\item $\{X_i>\VaR_{\alpha}(X_i)\}$, $i\in [n]$, are mutually exclusive;
	\item $(X_1,\dots,X_n)$ are $(n \alpha)$-concentrated. 
\end{enumerate}
 For a class $\rho$ of risk measures $\rho_\alpha$ decreasing in $\alpha$,  we say that  $\rho$ is \emph{non-flat from the left} at $(\alpha,X)$ if   $\rho_{\beta}(X)>\rho_{\alpha}(X)$ for all $\beta \in(0, \alpha)$, and 
 $\rho$ is \emph{left continuous} at $(\alpha,X)$ if
$\alpha\mapsto \rho_\alpha(X)$ is left continuous.


	\begin{remark}\label{rem:alpha_CE}
		For any given $X\in L^0$, if $\VaR$ is non-flat from the left at $(n\alpha,X)$, then there exists $\alpha$-CE random vector $\mathbf X \in \X^n$
    		such that $\sum_{i=1}^n X_i=X$.  For instance,  let $A=\left\{X>\operatorname{VaR}_{n \alpha}(X)\right\}$. As VaR is non-flat from the left at $(n\alpha,X)$, we have $\mathbb{P}(A)=n \alpha$. Let $(A_1, \ldots, A_n)$ be a partition of $A$ with $\mathbb{P}\left(A_i\right)=\alpha$ for $i\in[n]$. Also, let $X_i=(X-m) {\bf 1}_{A_i}$ for $i\in[n-1]$ and $X_n=(X-m) \mathbf{1}_{\left\{A_n \cup A^c\right\}}+m$
where $m=\operatorname{VaR}_{n \alpha}(X)$ is a constant. It follows that $\sum_{i=1}^n X_i=X$, and it is clear that $\mathbf{X}=$ $\left(X_1, \ldots, X_n\right)$ is an $\alpha$-CE model; such a construction is essentially the one in \citet[Theorem 2]{ELW18}.  More generally, we give a sufficient condition for $\mathbf X$ to satisfy the $\alpha$-CE model. A random vector $(X,Y)$ is said to be {counter-monotonic} if $(X,-Y)$ is comonotonic. 
		If  each pair $(X_i,X_j)$ is counter-monotonic for $i\ne j$,
		and  for each $i\in [n]$, $\p(X_i
		>\VaR_\alpha(X_i))=\alpha$ and $\VaR_\alpha(X_i)=\essinf(X_i)$,
 then $\mathbf X$ follows an $\alpha$-CE model.  For recent results on pairwise counter-monotonicity, see \cite{LLW23}.
	\end{remark}

	In  the next result, we summarize several dependence structures that correspond to special values $0$, $1$ and $n$ of  ${\rm DQ}^{\VaR}_\alpha$ and 
	the special values $0$ and $1$ of ${\rm DQ}^{\ES}_\alpha$.  
\begin{theorem}\label{th:var-01n}
 For $\alpha \in (0,1)$ and $n\ge 2$,  the following  hold:
  \begin{enumerate}[(i)]
\item   $\left\{\mathrm{DQ}^{\VaR}_{\alpha}(\mathbf{X}) \mid\mathbf{X} \in \X^n\right\}=[0,\min\{n,1/\alpha\}]$ and $\left\{\mathrm{DQ}^{\ES}_{\alpha}(\mathbf{X})\mid\mathbf{X} \in \X^n\right\}=[0,1]$.
  \item  For $\rho$ being $\VaR$ or $\ES$, $ {\rm DQ}^{\rho}_\alpha (\mathbf X)=0$ 
  if and only if $\sum_{i=1}^ n X_i\le \sum_{i=1}^ n \rho_\alpha(X_i)$ a.s.
In case $\sum_{i=1}^ n X_i$ is a constant, $ {\rm DQ}^{\VaR}_\alpha (\mathbf X)=0$ if $\alpha <1/n$ and  $ {\rm DQ}^{\ES}_\alpha (\mathbf X)=0$.
  \item For $\rho$ being $\VaR$ or $\ES$, if $\mathbf X$ is $\alpha$-concentrated,  then  $ {\rm DQ}^{\rho}_\alpha (\mathbf X)\le 1$.
  If, in addition, $\rho$ is continuous and non-flat from the left at $(\alpha,\sum_{i=1}^n X_i)$, then $ {\rm DQ}^{\rho}_\alpha (\mathbf X)=1$. 
    \item If $\alpha < 1/n$ and $\mathbf X$ has an $\alpha$-CE model, then 
 $ {\rm DQ}^{\VaR}_\alpha (\mathbf X)=n$ and $ {\rm DQ}^{\ES}_{n\alpha} (\mathbf X)=1$.
 
  \end{enumerate}
 \end{theorem}
\begin{proof}
	
(i) We first prove the case of $\VaR$. By Corollary 1 of \cite{ELW18},  we have
$$ \VaR_{n\alpha} \left(\sum_{i=1}^n X_i\right) \le \sum_{i=1}^n \VaR_{\alpha}(X_i),  $$
which implies 
 $\alpha^{*} \le n \alpha$,  and hence $\mathrm{DQ}^{\VaR}_{\alpha}(\mathbf{X})\le n$. By definition, $\alpha^* \in [0,1]$, and hence $0\le \mathrm{DQ}^{\VaR}_{\alpha}(\mathbf{X})\le 1/\alpha$. To summarize, $ \left\{\mathrm{DQ}^{\VaR}_{\alpha}(\mathbf{X}) \mid\mathbf{X} \in \X^n\right\}\subseteq [0,\min\{n,1/\alpha\}]$.

Next, we show that every point in the interval $ [0,\min\{n,1/\alpha\}]$ is attainable by $\mathrm{DQ}^\VaR_\alpha$. 
Take any $\mathbf X \in \X^n$
and let $a=\mathrm{DQ}^{\VaR}_\alpha (\mathbf X)$.
Since $\mathrm{DQ}^{\VaR}_\alpha$ satisfies [LI], we can replace each component $  X_i$ of $\mathbf X$ with $X_i - \VaR_\alpha(X_i)$ for $i\in [n]$. 
Hence, it is safe  to assume that  $\VaR_\alpha$ of each component of $\mathbf X$  is $0$. 
Let $\mathbf{Z}= \mathbf{X}\id_{A}$ where $A \in \mathcal F$ is independent of $\mathbf{X}$ and $\p(A)=p\in (0,1)$.
Since the mapping $F\mapsto \VaR_\alpha(X)$ where $X\sim F$ has convex level sets (e.g., \cite{G11}), 
$\VaR_\alpha$ of each component of $\mathbf Z$ is $0$. 
 By \eqref{eq:superquantile}, we have
$$
\begin{aligned}
\mathrm{DQ}^{\VaR}_\alpha (\mathbf Z)
 =\frac{1}{\alpha} \p\left(\sum_{i=1}^n Z_i> 0 \right) 
&=\frac{p }{\alpha}  \p\left(\sum_{i=1}^n X_i >  0 \right)  =p \mathrm{DQ}^{\VaR}_\alpha (\mathbf X).
\end{aligned}$$
Since $p\in (0,1)$ is arbitrary, any point in $[0,a]$ belongs to the range of $\mathrm{DQ}^\VaR_\alpha$. 
To complete the proof, it suffices to construct $\mathbf X$ such that $\mathrm{DQ}^\VaR_\alpha(\mathbf X)=\min\{n,1/\alpha\}$.

In case $\alpha\ge 1/n$, let $\mathbf X$ follow an $n$-dimensional multinomial distribution with parameters $(1/n,\dots,1/n)$. 
It is clear that $\sum_{i=1}^n X_i=1$. 
Since $\alpha \ge 1/n$, then $\VaR_\alpha(X_i)=0$. In this case, by \eqref{eq:superquantile}, $\mathrm{DQ}^\VaR_\alpha(\mathbf X)=1/\alpha$. 
In case  $\alpha< 1/ n$, we can find $\mathbf X$ satisfying  $\mathrm{DQ}^\VaR_\alpha(\mathbf X)=n$, which is constructed in part  (iv) of the proof  below.

Next, we prove  the case of   $\ES$.
Since ES satisfies [SA], the range of $\mathrm{DQ}^\ES_\alpha $ is contained in $[0,1]$. 
Take any $t\in [0,2]$, and let each of $X_1$ and $X_2$ follow a uniform distribution on $[-1,1]$ such that $X_1+X_2$ is uniformly distributed on $[-t,t]$. The existence of such $(X_1,X_2)$ is shown by Theorem 3.1 of \cite{WW16}.  Let $X_i=0$ for $i= 3,\dots,n$.
We can easily compute
$
\ES_\alpha(X_1)=\ES_\alpha(X_2)=1-\alpha
$
and 
$
\ES_\beta(X_1 +X_2)=t(1-\beta).
$
Hence, 
$$\mathrm{DQ}^\ES_\alpha(X_1,\dots,X_n) = \frac 1 \alpha 
\inf \{\beta \in (0,1): t(1-\beta)\le 2-2\alpha\} = \frac 1 \alpha  \left(1-\frac{2-2\alpha}{t}\right)_+.
$$
For letting $t$ vary in $[0,2]$, we get that every point in $[0,1]$ is attained by $\mathrm{DQ}^\ES_\alpha$.

			(ii)  The first part follows directly from Theorem 2 (i) of \cite{HLW22}.  In particular,  if $\sum_{i=1}^n X_i$ is a constant, we have  $\VaR_{0}\left(\sum_{i=1}^n X_i\right)=\VaR_{n\alpha}\left(\sum_{i=1}^n X_i\right)\leq\sum_{i=1}^n\VaR_{\alpha}( X_i)$ for $\alpha<1/n$, and $\ES_{0}\left(\sum_{i=1}^n X_i\right)=\ES_{\alpha}\left(\sum_{i=1}^n X_i\right)\leq\sum_{i=1}^n\ES_{\alpha}( X_i)$. Thus, we have $\mathrm{DQ}_\alpha^\ES(\mathbf{X})=0$ if $\alpha<1/n$ and  $\mathrm{DQ}_\alpha^\ES(\mathbf{X})=0.$
   
			(iii) By Theorem 6 in \cite{WZ20},  if  $\mathbf{X}$ is $\alpha$-concentrated,  we have
			$$
			\VaR_{\alpha}\left(\sum_{i=1}^{n} X_{i}\right) \le \sum_{i=1}^{n} \VaR_{\alpha}\left(X_{i}\right),
			$$ which implies $\alpha^*\leq\alpha$ and  then $\mathrm{DQ}_\alpha^\VaR(\mathbf X)\leq1.$  Further, as $\VaR$ is continuous and non-flat from the left at $(\alpha, \sum_{i=1}^{n} X_{i})$, by Theorem 6 in \cite{WZ20},   the inequality above is  an equality. Thus, we have $\alpha^{*}=\alpha$, which leads to ${\rm DQ}_{\alpha}^{\VaR}(\mathbf{X})=1$. Moreover, from Theorem 5 of \cite{WZ20}, we know that $\ES_{\alpha}\left(\sum_{i=1}^{n} X_{i}\right)=\sum_{i=1}^{n} \ES_{\alpha}\left(X_{i}\right)$ if $\left(X_{1}, \dots,  X_{n}\right)$ is $\alpha$-concentrated. Combining with the fact that $\ES_\alpha(\sum_{i=1}^n X_i)$ is  non-flat from left at $(\alpha,\mathbf X)$, we have $ {\rm DQ}^{\ES}_\alpha (\mathbf X)= 1$.
			
			(iv) 
			As $X_1,\dots,X_n$ are $(n \alpha)$-concentrated, there exists an event $B$ such that $B$ is a tail event for all $X_i$ and $\p(B)=n \alpha$. Let $B_i=\{X_i>\VaR_{\alpha}(X_i)\}$.
			By Lemma A.3 of \cite{WZ20}, we have $\{X_i>\VaR_{n\alpha}(X_i)\}\subseteq B$. As $\VaR_{\alpha}(X_i)\ge \VaR_{n \alpha}(X_i)$,  it gives $B_i\subseteq B$ for all $i\in [n]$.
			From $\mathbb{P}(X_i\geq \VaR_\alpha(X_i))\geq n\alpha$,  we know that  $X_i(\omega)\ge \VaR_{\alpha}(X_i)$ for all $\omega \in B$. Further, as $B_1,\dots,B_n$ are mutually exclusive, we have $X_i(\omega)>\VaR_{\alpha}(X_i)$ and $X_j(\omega)=\VaR_{\alpha}(X_j)$ for all $\omega \in B_i$ and $j \neq i$.
			Hence, for all $\omega \in \bigcup_{i=1}^n B$, we have  $\sum_{i=1}^n X_i(\omega)> \sum_{i=1}^n \VaR_{\alpha}(X_i)$
			while $\sum_{i=1}^n X_i(\omega)\le \sum_{i=1}^n \VaR_{\alpha}(X_i)$  for $\omega \in \left(\bigcup_{i=1}^n A_i\right)^c=\bigcap_{i=1}^n A_i^c$. Therefore, if $\alpha<1/n$,
			$$\p\left(\sum_{i=1}^n X_i>\sum_{i=1}^n \VaR_{\alpha} (X_i)\right)=\p\left(\bigcup_{i=1}^n B_i\right)=\sum_{i=1}^n \p(B_i)=n\alpha.$$ By \eqref{eq:superquantile}, we have ${\rm 
				DQ}_{\alpha}^{\VaR}(\mathbf X)=n$.
			
			For the case of ES,  as $X_1,\dots,X_n$ are $(n \alpha)$-concentrated, by Theorem 5 of  \cite{WZ20},  we have $\ES_{n\alpha}\left(\sum_{i=1}^{n} X_{i}\right)=\sum_{i=1}^{n} \ES_{n\alpha}\left(X_{i}\right)$. Together with the fact that $\beta\mapsto \ES_\beta \left(\sum_{i=1}^{n} X_{i}\right)$ is strictly decreasing at $\beta=n\alpha$, we get  that ${\rm DQ}_{n\alpha}^{\ES}(\mathbf X)=1$.
		\end{proof}

	Note that comonotonicity is stronger than $\alpha$-concentration, and hence it is a sufficient condition for
 (iii) in Theorem \ref{th:var-01n}   replacing $\alpha$-concentration. 
 
 In summary, both ${\rm DQ}^{\VaR}_\alpha$ and ${\rm DQ}^{\ES}_\alpha$
take values on a bounded interval. In contrast, the diversification ratio ${\rm DR}^{\VaR_\alpha}$ is unbounded,  and  ${\rm DR}^{\ES_\alpha}$  is  bounded above by $1$ only  when the ES of the total risk is non-negative.  The continuous ranges of DQs also give more information on diversification. Moreover, similarly to the continuity axiom of preferences (e.g., \cite{FS16}), a bounded interval can  provide mathematical convenience for applications.
The values of DQs are simple to interpret.
To be specific,
for ${\rm DQ}^{\VaR}_\alpha$, its value is $0$ if there is a very good hedge in the sense of Theorem \ref{th:var-01n}  (ii);
its value is $1$ if there is strong positive dependence such as comonotonicity, and
its value is $n$ if there is strong negative dependence conditional on the tail event. For ${\rm DQ}^{\ES}_\alpha$,
 its value is $0$ if there is a very good hedge in the sense of  Theorem \ref{th:var-01n} (ii) and
its value is $1$ if there is strong positive dependence such as comonotonicity or $\alpha$-concentration.	

\section{Diversification for large portfolios}
\label{large-port}
In this section, we will focus on the asymptotic behavior of DQ for large portfolios.  First, since the independent portfolio is widely recognized as an effectively diversified portfolio, we anticipate that ${\rm DQ}$ for this type of portfolio would be close to zero as $n$ tends to $\infty$. 


{\begin{theorem}\label{pro:ind}
Let  $X_1, X_2, \dots $ be a sequence of uncorrelated random variables in $L^2$. Assume $\sup_{i\in \mathbb N} \var(X_i) < \infty$ and $\inf_{i\in \mathbb N} \{\rho_{\alpha}(X_i)-\E[X_i]\} >0$. 
 For $\alpha\in(0,1)$ and  $\rho$ being $\VaR$ or $\ES$, 
\begin{equation}\label{eq:indep}
\lim_{n \to \infty} \mathrm{DQ}_\alpha^{\rho}(X_1, \dots, X_n)=0.
\end{equation}
\end{theorem}
\begin{proof}
Let $\mathbf X_n=(X_1, \dots, X_n)$ and $S_n=\sum_{i=1}^n X_i$. As $\mathrm{DQ}^\rho_\alpha$ is location invariant, we  can assume that $\E[X_i]=0$ for $i=1, 2, \dots$. Hence, by the $L^2$-Law of Large Numbers in the form of \citet[Theorem 2.2.3]{D19}, we have $S_n/n \buildrel L^2 \over \rightarrow 0$. (In fact, $L^1$ convergence is sufficient to prove our result.)

We first prove the case of VaR. Note that $S_n/n \buildrel L^2 \over \rightarrow 0$ implies $\lim_{n \to \infty} \p(S_n/n>x)=0$ for all $x>0$.
Let  $\epsilon = \inf_{i\in \mathbb N} \{\rho_{\alpha}(X_i)-\E[X_i]\} $.  As $\VaR_{\alpha}(X_i)>\epsilon$, $i=1,2, \dots$, we have
$$\p\left(S_n>\sum_{i=1}^n \VaR_\alpha(X_i)\right)\le \p\left(S_n/n>\epsilon\right)\to 0.$$
Thus, $\lim_{n\to \infty} \p(S_n>\sum_{i=1}^n \VaR_\alpha(X_i))=0$.
By   \eqref{eq:superquantile}, we have 
\begin{align*}
\lim_{n\to \infty}\mathrm{DQ}^{\VaR}(\mathbf X_n)=\lim_{n \to \infty}\frac{1}{\alpha}\p\left(S_n>\sum_{i=1}^n \VaR_\alpha(X_i)\right)=0.
\end{align*}

Next, we prove the case of $\ES$.  As $\ES$ is a convex distortion  risk measure, $\ES$ is $L^1$-continuous (see \citet[Corollary 7.10]{R13}). Further, since  $\ES_\beta(0)=0$,  we have $\ES_\beta(S_n/n) \to 0$  as $n\to\infty$ for all $\beta \in (0,1)$. As a result, for every $\beta \in (0,1)$, there exists $N_\beta$ such that $\ES_\beta(S_n/n)<\epsilon$ for all $n>N_\beta$. Therefore, we have  
$$\alpha^*=\inf\left\{\beta \in (0,1):\ES_\beta(S_n)\le \sum_{i=1}^n \ES_\alpha(X_i)\right\}\le \inf\left\{\beta \in (0,1):\ES_\beta(S_n/n)\le \epsilon)\right\} \to 0$$
as $n \to \infty$.
Hence, we have $\mathrm{DQ}^\ES_\alpha(\mathbf X_n)=\alpha^*/\alpha \to 0$ as $n \to \infty$.
\end{proof}
Note that  Theorem \ref{pro:ind} does not imply that all independent portfolios are good hedges, because \eqref{eq:indep}  holds under some assumptions. In  case  the components of the portfolio  have very heavy tails,   DQ based on VaR 
can be close to $n$ even if the individual losses are iid, as we will see in Theorem \ref{thm:MRV} below.

\begin{remark}\label{rem:iid}
In the special case that $X_1, X_2,  \dots$ are iid,  Theorem \ref{pro:ind} implies that,  if $\rho_\alpha(X_1)>\E[X_1]$, we have
$$
\lim_{n \to \infty} \mathrm{DQ}_\alpha^{\rho}(X_1, \dots, X_n)=0
$$  for $\rho$ being $\VaR$ or $\ES$.
\end{remark}}

Next, we focus on portfolios with exchangeable components, which may represent a homogeneous subgroup of assets from a large asset pool. 
An infinite sequence of random variables $X_1, X_2, \dots$  is said to be exchangeable if $(X_1, \dots, X_n)\laweq (X_{\pi(1)}, \dots, X_{\pi(n)})$ for all $n\ge 2$ and $\pi \in {\mathfrak S_n}$, where ${\mathfrak S_n}$ is the set of permutations of $[n]$.
Exchangeability is closely related to iid sequence of random variables  due to de Finetti's theorem, which says that any infinite exchangeable sequence is conditionally iid.
However, for the exchangeable portfolio, the value of DQ does not necessarily converge to $0$ as $n$ goes to infinity.
 By  the Birkhoff–Khinchin theorem (see \cite{AK49}),
if $\E[\vert X_1\vert]<\infty$, we have $ \sum_{i=1}^n X_i/n \to \E[X_1|\mathcal G]$ a.s.~for some sub-$\sigma$-algebra  $\mathcal G \subseteq \mathcal F$. By \eqref{eq:superquantile}, we get
$$\mathrm{DQ}^{\VaR}_\alpha(X_1, \dots, X_n) \to \frac{1-F\left(\VaR_{\alpha}(X_1)\right)}{\alpha} ~~~\mbox{as}~~~ n \to \infty,
$$
and 
$$\mathrm{DQ}^{\ES}_\alpha(X_1, \dots, X_n) \to \frac{1-\widetilde F\left(\ES_{\alpha}(X_1)\right) }{\alpha}~~~\mbox{as}~~~ n \to \infty,
$$ 
where $F$ is the distribution of $\E[X_1|\mathcal G]$   and $\widetilde F$ is  the superquantile transform of $F$.


The above formulas depend on $\mathcal G$ which may not be explicit. In the next proposition,  we  derive an upper bound  on the limit.
\begin{proposition}\label{prop:exchangeable}
   Let $X_1$, $X_2, \dots$  be a sequence of exchangeable random variables in $L^2$. Denote by  $\mu=\E[X_1]$, $\sigma^2=\var(X_1)$ and $r=\mathrm{corr}(X_1, X_2)$. For $\alpha\in(0,1)$ and $\rho$ being $\VaR$ or $\ES$, if $\rho_\alpha(X_1)>\mu$,  then
\begin{align} 
\label{eq:r1-exchange} \lim_{n\to \infty}
\mathrm{DQ}^{\rho}_\alpha(X_1, \dots, X_n) \le  \frac{1}{\alpha}\frac{r \sigma^2}{r \sigma^2+(\rho_\alpha(X_1)-\mu)^2}
.\end{align}
\end{proposition}
\begin{proof}
Let $S_n=\sum_{i=1}^n X_i$. As $(X_1, \dots, X_n)$ is exchangeable, we have $\E[S_n]=n\mu$ and $\var(S_n)=(n+n(n-1)r)\sigma^2$. The mean and variance of $S_n$ imply  the bound 
$$\rho_\beta(S_n)\le n\mu+\sigma\sqrt{n+n(n-1)r}\sqrt{\frac{1-\beta}{\beta}}$$
for all $\beta \in (0,1)$; see Table 1 of \cite{LSWY18}.
As a result, we have
\begin{align*}
{\rm DQ}^\rho_\alpha(X_1, \dots, X_n)&\le \frac{1}{\alpha}\inf\left\{\beta\in (0,1): n\mu+\sigma\sqrt{n+n(n-1)r}\sqrt{\frac{1-\beta}{\beta}}\le n \rho_\alpha(X_1)\right\}\\&= \frac{1}{\alpha}\frac{\frac{1+(n-1)r}{n}\sigma^2}{\frac{1+(n-1)r}{n}\sigma^2+(\rho_\alpha(X_1)-\mu)^2}.
\end{align*}
Sending $n \to\infty$, we get the desired result.
\end{proof}
The upper bound \eqref{eq:r1-exchange} on $\lim_{n\to \infty}
\mathrm{DQ}^{\rho}_\alpha(X_1, \dots, X_n) $ in Proposition \ref{prop:exchangeable} decreases as the correlation $r$ between  assets decreases. Intuitively, this means that less positive dependence leads to greater diversification.   In particular, if $r\downarrow 0$, then $\lim_{n\to \infty}
\mathrm{DQ}^{\rho}_\alpha(X_1, \dots, X_n) \to 0$. The upper bound \eqref{eq:r1-exchange} holds true also without exchangeability, as long as the average of the bivariate correlations of assets converges to $r$ and all assets are identically distributed.


	

	

	\section{Elliptical   models}\label{sec:4}
	
	The most commonly used classes of multivariate distributions are the elliptical models which  include the multivariate normal    and   t-distributions as  special cases.  For a general treatment of elliptical models  in risk management, see \cite{MFE15}.  In this section, we   study DQs based on VaR and ES for elliptical models.  %
	\subsection{Explicit formulas for DQ}
	\label{sec:ellip} 
	
	A random vector $\mathbf{X}$ is \emph{elliptically distributed} if
	its characteristic function can be written as
	$$
	\begin{aligned}
		\psi(\mathbf{t}) =\mathbb{E}\left[\exp \left(\texttt{i} \mathbf{t}^\top \mathbf{X}\right)\right] & 
		=\exp \left(\texttt{i} \mathbf{t}^\top \boldsymbol{\mu}\right) \tau\left(\mathbf{t}^\top \Sigma \mathbf{t}\right),
	\end{aligned}
	$$
	for some  $\boldsymbol{\mu}\in \mathbb{R}^{n}$, positive semi-definite matrix $ \Sigma\in \R^{n\times n}$,
	and $\tau: \mathbb{R}_{+} \rightarrow \mathbb{R}$ called the characteristic generator.
	We denote this distribution
	by $ \mathrm{E}_{n}(\boldsymbol{\mu}, \Sigma, \tau).
	$ We will assume that $\Sigma$ is not a  matrix of zeros.
	Each marginal distribution of an elliptical distribution is a one-dimensional elliptical distribution with the same characteristic generator.
	The most common examples of elliptical distributions are normal and t-distributions.
	An $n$-dimensional t-distribution  $\mathrm  t(\nu,\boldsymbol{\mu},\Sigma) $ with $\nu>0$ has density function $f$ given by (if $|\Sigma| > 0$)
	$$f(\mathbf x)= {{\frac {\Gamma \left((\nu +n)/2\right)}{\Gamma (\nu /2)\nu ^{n/2}\pi ^{n/2}\left|{{\Sigma }}\right|^{1/2}}}\left(1+{\frac {1}{\nu }}({\mathbf {x} }-{\boldsymbol {\mu }})^{\top}{ {\Sigma }}^{-1}({\mathbf {x} }-{\boldsymbol {\mu }})\right)^{-(\nu +n)/2}},$$
	where $\Gamma$ is the gamma function and $|\Sigma|$ is the determinant of the  dispersion matrix $\Sigma$. 

	We   remind the reader that for elliptical models, VaR and ES behave very similarly. For instance, $\VaR_\alpha$ is subadditive for $\alpha \in (0,1/2)$ in this setting; see \cite[Theorem 8.28]{MFE15}. Moreover, for $\mathbf X \sim  \mathrm{E}_{n}(\boldsymbol{\mu}, \Sigma, \tau)$ and $\mathbf a\in \R^n$, both $\VaR_\alpha(\mathbf a^\top \mathbf X)$ and $\ES_\alpha(\mathbf a^\top \mathbf X)$ have the form $y \sqrt{\mathbf{a}^\top \Sigma \mathbf a}  + \mathbf a^\top \boldsymbol \mu$ for some constant $y$ being $y^{\VaR}_\alpha:=\VaR_\alpha(Y)$ or $y^\ES_\alpha:=\ES_\alpha(Y)$ where $Y\sim \mathrm{E}_1(0,1,\tau)$.
	As a consequence, the behaviour of DQ based on VaR  is similar to that based on ES, except for the case of infinite mean. 

	For  a positive semi-definite matrix $\Sigma$,
	we write $\Sigma=(\sigma_{ij})_{n\times n}$, $\sigma_i^2=\sigma_{ii}$, and $\boldsymbol \sigma=(\sigma_1,\dots,\sigma_n)$, and
	define the constant
	\begin{equation}\label{eq:k}k_\Sigma= \frac {\sum_{i=1}^n 
			\left(\mathbf{e}^\top_i \Sigma \mathbf{e}_i \right)^{1/2}} {\left( \mathbf{1}^\top \Sigma \mathbf{1}\right)^{1/2}  } 
	=\frac{\sum_{i=1}^n\sigma_{i} }{ \left(\sum_{i, j}^n \sigma_{ij}\right)^{1/2} }
	\in [1,\infty),\end{equation}
where   $\mathbf{1}=(1,\dots,1)\in\R^n$ and  $ \mathbf e_{1},\dots, \mathbf e_{n}$ are the column vectors of the $n\times n$ identity matrix $I_n$.
Moreover, $k_\Sigma = 1$ if and only if $\Sigma =\boldsymbol \sigma \boldsymbol \sigma^\top  $, which means that $\mathbf X\sim \mathrm{E}_n( \boldsymbol{\mu}, \Sigma,\tau)$ is comonotonic.

Explicit formulas  and the limiting behavior of DQs based on VaR and ES for elliptical models   are given by the following few results. 

\begin{proposition}
	\label{prop:comp_Dvar}
	Suppose that $\mathbf X \sim  \mathrm{E}_{n}(\boldsymbol{\mu}, \Sigma, \tau)$.  We have, for $\alpha \in (0,1)$,
	$${\rm DQ}_\alpha^{\VaR}(\mathbf X)
	=
		\frac{1-F  (k_\Sigma \VaR_{\alpha}(Y)  )}{\alpha}
		~~\text{and}~~
		{\rm DQ}_\alpha^{\ES}(\mathbf X)=
		\frac{1- \widetilde  F (k_\Sigma \ES_{\alpha}(Y)  )}{\alpha},		
		$$
		where $ Y \sim \mathrm{E}_1(0,1,\tau)$ with  distribution function  $F$, and  $ \widetilde F$ is the superquantile transform of  $F$ in \eqref{eq:superquantile}.
		Moreover, 
		\begin{enumerate}[(i)]   
			\item   $\alpha \mapsto {\rm DQ}_\alpha^{\VaR}(\mathbf X)$ takes value in  $[0,1]$ on $(0,1/2]$ and it takes value in   $[1,2]$ on $(1/2,1)$; \item $k_\Sigma \mapsto {\rm DQ}_\alpha^{\VaR}(\mathbf X)$ is decreasing for $\alpha \in (0,1/2]$ and increasing for $\alpha \in (1/2,1)$;
			\item $k_\Sigma \mapsto {\rm DQ}_\alpha^{\ES}(\mathbf X)$ is decreasing for $\alpha \in (0,1)$.
		\end{enumerate}
	\end{proposition} 
	\begin{proof}
		We first consider the case of VaR. Since  $\mathbf{X} \sim \mathrm{E}_{n}(\boldsymbol{\mu}, \Sigma, \tau)$,  the linear  structure of ellipitical distributions gives    $\sum_{i=1}^n X_i \sim  \mathrm{E}_{1}( \mathbf I^\top \boldsymbol{\mu}, \mathbf I^\top\Sigma \mathbf I, \tau)$.  That is, $\sum_{i=1}^n X_i \laweq \sum_{i=1}^n \mu_i+\Vert \mathbf{1}^\top A\Vert_2 Y$, where $A$ is the Cholesky decomposition of $\Sigma$.
		Also,  we have $ \VaR_\alpha(X_i)=\mu_i+\Vert  \mathbf{e}^\top_iA\Vert_2 \VaR_\alpha (Y).$
		By \eqref{eq:superquantile},
		$$
		\begin{aligned}\mathrm{DQ}^\VaR_\alpha(\mathbf X)&=\frac{1}{\alpha}\p\left(\sum_{i=1}^n X_i>\sum_i^n \mu_i+\Vert  \mathbf{e}^\top_i A\Vert_2 \VaR_\alpha (Y)\right)\\
			&=\frac{1}{\alpha}\p\left(\sum_{i=1}^n \mu_i+\Vert \mathbf{1}^\top A\Vert_2 Y>\sum_i^n \mu_i+\Vert  \mathbf{e}^\top_i A\Vert_2 \VaR_\alpha (Y)\right)=\frac{1-F(k_\Sigma \VaR_{\alpha} (Y))}{\alpha}.
		\end{aligned}$$
		By replacing $\VaR$ with $\ES$ and $\sum_{i=1}^n X_i$ with $\ES_U(\sum_{i=1}^n X_i)$, we can get the first formula of $\mathrm{DQ}^\ES_\alpha(\mathbf X)$. 
		\begin{enumerate}[(i)]
			\item  For  $\alpha\in(0,1/2]$, we have $\VaR_{\alpha} \left(Y\right) \le k_{\Sigma}\VaR_{\alpha}(Y)$ and $1-\alpha\le F(k_\Sigma \VaR_\alpha(Y)) \le 1$. Hence, $0\le \mathrm{DQ}^\VaR_\alpha(\mathbf X)\le 1$.
			
			For  $\alpha\in(1/2,1)$, $\VaR_{\alpha} \left(Y\right) \ge k_{\Sigma}\VaR_{\alpha}(Y)$ and $\alpha\le 1- F(k_\Sigma \VaR_\alpha(Y))\le 1$. Hence, $1\le\mathrm{DQ}^\VaR_\alpha(\mathbf X)\le 1/\alpha\le 2$.
			\item  If $\alpha\in(0,1/2]$,  then $\VaR_\alpha (Y)\geq0$, and thus $\mathrm{DQ}^\VaR_\alpha(\mathbf X)$ decreases in $k_{\Sigma}$. If $\alpha\in(1/2,1)$,  then $\VaR_\alpha (Y)\leq0$, and thus $\mathrm{DQ}^\VaR_\alpha(\mathbf X)$ increases in $k_{\Sigma}$.
			\item For $\alpha \in (0,1)$, $\ES_\alpha (Y)\ge 0$. Hence, $\mathrm{DQ}^\ES_\alpha(\mathbf X)$ increases in $k_{\Sigma}$.
			\qedhere
		\end{enumerate}  
	\end{proof}

	In the discussions below, we will assume $\alpha \in (0,1/2)$, which is the most common setting in risk management.
	In Proposition \ref{prop:comp_Dvar},  we see that, for $\alpha \in (0,1/2)$, 
	${\rm DQ}_\alpha^{\VaR}(\mathbf X)\in [0,1]$.
	This is in contrast to Theorem \ref{th:var-01n}, where the range of ${\rm DQ}_\alpha^{\VaR}$ is $[0,n]$ instead of $[0,1]$, when we do not restrict to elliptical models. This phenomenon should not be surprising, because, as we mentioned before, $\VaR_\alpha$ for $\alpha \in (0,1/2)$ is similar to $\ES_\alpha$ for elliptical models,
	and  ${\rm DQ}_\alpha^{\ES}$ has range $[0,1]$.

	In case  $ Y \sim \mathrm{E}_1(0,1,\tau)$ has a positive density on $\R$, we can see from Proposition \ref{prop:comp_Dvar} that $   {\rm DQ}_\alpha^{\VaR}(\mathbf X)=1$ if and only if $k_\Sigma=1$ (i.e., $\mathbf X$ is comonotonic) or $\VaR_\alpha(Y)=0$ (i.e., $\alpha=1/2$).
	Similarly,  
	$   {\rm DQ}_\alpha^{\ES}(\mathbf X)=1$ if and only if $k_\Sigma=1$.

	In case the   elliptical distribution is asymptotically uncorrelated,
	we will see   that ${\rm DQ}_\alpha^{\VaR}(\mathbf X) \to 0$ and ${\rm DQ}_\alpha^{\ES}(\mathbf X) \to 0$ as $n\to \infty$.  This is consistent with our intuition that, if the individual risks are asymptotically uncorrelated, then full diversification can be achieved asymptotically, thus the diversification index goes to $0$.
	%
	The value  $\mathrm{AC}_{\Sigma}= \sum_{i, j}^n \sigma_{ij}/(\sum_{i=1}^n\sigma_i)^2 = 1/k_\Sigma^2$  will be called the average correlation (AC) of $\Sigma$.
	
	%
	%
	%

	\begin{proposition}\label{cor:VaR} 
		Suppose that $\mathbf X \sim  \mathrm{E}_{n}(\boldsymbol{\mu}, \Sigma, \tau)$.
		\begin{enumerate}[(i)]
			\item 
			Let  $Y \sim \mathrm{E}_1(0,1,\tau)$ and $f$ be the  density function of $Y$. We have
			\begin{equation}
				\label{eq:limit} \lim_{\alpha \downarrow 0} {\rm DQ}^{\VaR}_{\alpha}(\mathbf X)=
				\lim_{x \to \infty}k_\Sigma \frac{f(k_\Sigma x)}{f(x)} \mbox{~~~if $ \VaR_0(Y)=\infty  $ and the limit exists},
			\end{equation}
			and $\lim_{\alpha \downarrow 0} {\rm DQ}^{\VaR}_{\alpha}(\mathbf X)=0$ if $ \VaR_0(Y)<\infty$.
			\item 
			If $\lim_{n\rightarrow\infty}\mathrm{AC}_{\Sigma}= 0$,  then  $$\lim_{n \to \infty}{\rm DQ}_\alpha^{\VaR}(\mathbf X)  = \lim_{n \to \infty}{\rm DQ}_\beta ^{\ES}(\mathbf X) =0$$  for $\alpha\in (0,1/2)$ and $\beta \in (0,1)$.
		\end{enumerate}
	\end{proposition}
	
	\begin{proof}
		(i)
		If $\VaR_0(Y)<\infty$, then $\VaR_0(Y)\le k_\Sigma\VaR_0(Y)$ as $k_{\Sigma}\ge 1$. Hence,  ${\rm DQ}_0^{\VaR}(\mathbf X)=0$.
		If $\VaR_0(Y)=\infty$, then  $\VaR_0(Y)>k_\Sigma\VaR_\alpha(Y)$ for $\alpha>0$. Therefore,
		$$
		\begin{aligned}
			\lim_{\alpha \to 0} {\rm DQ}^{\VaR}_{\alpha}(\mathbf X)&=\lim_{\alpha \to 0}  \frac{1-F\left(k_\Sigma \VaR_{\alpha}(Y))\right)}{\alpha} =\lim_{\alpha \to 0}  k_\Sigma\frac{f\left(k_\Sigma \VaR_{\alpha}(Y))\right)}{f( \VaR_{\alpha}(Y)))}  =\lim_{x \to \infty} k_\Sigma \frac{f\left(k_\Sigma x\right)}{f( x)},
		\end{aligned}
		$$
		and we get the desired result.

		(ii) 
		We only show the proof of   ${\rm DQ}_\alpha^{\VaR}$ as the result for ${\rm DQ}_\beta^{\ES}$ can be obtained along the same analogy.  By Proposition \ref{prop:comp_Dvar}, it is clear that  $\mathrm{AC}_\Sigma \to \mathrm{DQ}_{\alpha}^\VaR(\mathbf X)$ is increasing for $\alpha \in (0,1/2)$ and $\mathrm{AC}_\Sigma \to \mathrm{DQ}_{\beta}^\ES(\mathbf X)$ is increasing for $\alpha \in (0,1)$.
		Moreover, if  $\mathrm{AC}_{\Sigma}$ goes to $0$ as $n\to\infty$,  we have
		$\lim_{n\rightarrow\infty} k_\Sigma=\infty
		$. Thus, we have ${\rm DQ}_\alpha^{\VaR}(\mathbf X) \to 0$ as $n\to \infty$ by Proposition \ref{prop:comp_Dvar}.
	\end{proof}
	
	Explicit formulas of \eqref{eq:limit} for normal and t-distributions  are  provided in Section \ref{sec:normal-t}.
	
	\begin{remark}\label{rem:ES-t}
		In general, we do not have a limiting result for $\mathrm {DQ}^{\ES}_\alpha$ in the form of Proposition \ref{cor:VaR} (i).
		If $\mathbf X\sim \mathrm{t}(\nu, \boldsymbol \mu,\Sigma)$ for $\nu>1$, then $\mathrm {DQ}^{\ES}_\alpha$   has the same limit as $\mathrm {DQ}^{\VaR}_\alpha$  in \eqref{eq:limit} as $\alpha\downarrow 0$ because $\VaR_\alpha(Y)/\ES_\alpha(Y)$ has a constant limit $(\nu-1)/\nu$ for a t-distributed $Y$ by the Karamata theorem; see Theorem A.7 of \cite{MFE15}.
	\end{remark}
	
	From the results above,  
	${\rm DQ}^{\VaR}_{\alpha}(\mathbf X)$  and ${\rm DQ}^{\ES}_{\alpha}(\mathbf X)$ depend on both $\tau$ and   $\alpha$.
	In sharp contrast, DR of a centered  elliptical distribution is always $1/k_{\Sigma}$, which ignores the shape of the distribution.
	More precisely, for $\mathbf X \sim  \mathrm{E}_{n}(\mathbf{0}, \Sigma, \tau)$ and $\alpha\in(0,1/2)$, we have
	\begin{equation}
		\label{eq:DRellip}
		{\rm DR}^{\VaR_\alpha}(\mathbf X)=\frac{\VaR_{\alpha}(\sum_{i=1}^n X_i)}{\sum_{i=1}^n \VaR_{\alpha}(X_i)}=\frac{ \left(\sum_{i, j}^n \sigma_{ij}\right)^{1/2}   \VaR_\alpha (Y) 
		}{\sum_{i=1}^n \sigma_i \VaR_\alpha (Y)} = \frac{1}{k_\Sigma},
	\end{equation}
	and similarly, $ {\rm DR}^{\ES_\alpha}(\mathbf X)= {1}/{k_\Sigma}$.
	Note that in this case, ${\rm DR}^{\VaR_\alpha}$ and ${\rm DR}^{\ES_\alpha}$  do not depend on $\tau$, $\alpha$ or whether the risk measure is VaR or ES. Indeed, DR based on variance or SD also has the same value $1/k_\Sigma$.
	
	For $\mathbf X \sim  \mathrm{E}_{n}(\boldsymbol{\mu}, \Sigma, \tau)$ with $\boldsymbol \mu \ne \mathbf 0$,
	${\rm DR}^{\VaR_\alpha}(\mathbf X)$ and ${\rm DR}^{\ES_\alpha}(\mathbf X)$ depend also on $\boldsymbol \mu$, which is arguably undesirable as it conflicts location invariance.
	Nevertheless,  $\lim_{\alpha \downarrow 0} {\rm DR}^{\VaR_\alpha}(\mathbf X) =1/k_\Sigma $  if $\VaR_{0} (Y) =\infty$ (i.e., the value taken by $Y$ is unbounded from above), and this limit does not depend on $\boldsymbol \mu$.
	On the other hand,  ${\rm DQ}^{\VaR}_{\alpha}(\mathbf X)$ has a limit in \eqref{eq:limit} which depends on both $k_\Sigma$ and $\tau$.
	The above observations suggest that DQ is more comprehensive than DR by utilizing the information on the shape of the distribution.

	A similar result to Proposition \ref{cor:VaR} (ii) holds for DR of centered elliptical distributions. More precisely,
	If  $\alpha\in (0,1/2)$, $\boldsymbol \mu=\mathbf 0$,  and $\lim_{n\rightarrow\infty}\mathrm{AC}_{\Sigma}= 0$, then  we have $\lim_{n \to \infty}{\rm DR}^{\VaR_\alpha}(\mathbf X) =0$ by \eqref{eq:DRellip}, and similarly, $\lim_{n \to \infty}{\rm DR}^{\ES_\alpha}(\mathbf X) =0$.
	These limits do not hold if $\boldsymbol \mu\ne \mathbf  0$.

	\subsection{Normal and t-distributions}
	\label{sec:normal-t}
	
	Next, we  take a close look at the two most important elliptical distributions used in finance and insurance, namely  the multivariate normal distribution
	and the multivariate t-distribution.
	The explicit formulas for DQ for these distributions are  available through the explicit formulas of VaR and ES; see Examples 2.14 and 2.15 of \cite{MFE15}.

	\cite{HLW22}  proposed   three   simple models  where the components  of portfolio vectors follow the  iid normal model, iid t-model and the common shock t-model, respectively, and showed that the diversification is the strongest according to DQ for the iid normal model and   the iid t-model has a smaller DQ than the common shock t-model. In contrast, DR reports a similar value for all three models; see their Section 5.2 for details.  Therefore, DQ  has the nice feature that it can  capture heavy tails and common shocks.

	We present some formulas and numerical results for correlated normal and t-models. 
	We focus our discussions mainly on $\mathrm{DQ}^{\VaR}_\alpha$
	as the case of $\mathrm{DQ}^{\ES}_\alpha$ is similar.
	We first compute the limit of DQ as $\alpha \downarrow 0$ according to  \eqref{eq:limit}.  By direct calculation,  
	\begin{equation}
		\label{eq:limit-n}
		\lim_{\alpha \downarrow 0} {\rm DQ}^{\VaR}_{\alpha}(\mathbf X) =\id_{\{k_\Sigma=1\}} 
		\mbox{~~~if $\mathbf X\sim \mathrm{N}(\boldsymbol \mu,\Sigma)$;}
	\end{equation} 
	\begin{equation}
		\label{eq:limit-t}
		\lim_{\alpha \downarrow 0} {\rm DQ}^{\VaR}_{\alpha}(\mathbf X)= k_\Sigma^{-\nu}   \mbox{~~~if $\mathbf X\sim\mathrm{t}(\nu, \boldsymbol \mu,\Sigma)$.}
	\end{equation}
	The above two values properly reflect the fact that the normal distribution is tail independent unless $k_\Sigma=1$ (i.e., comonotonic), whereas the t-distribution is tail dependent; see Examples 7.38 and 7.39 of \cite{MFE15}.
	DQ is able to capture this phenomenon well, by providing,  for $\alpha$ close to $0$, $ {\rm DQ}^{\VaR}_{\alpha} \approx 0$ (strong diversification) for normal distribution and $ {\rm DQ}^{\VaR}_{\alpha} \approx k_\Sigma^{-\nu}$ (moderate diversification for common choices of $\Sigma$ and $\nu$; see Figure \ref{fig:ellip}) for a t-distribution.
	On the other hand, DR of centered normal and t-distributions is always $1/k_{\Sigma}$, which fails to distinguish
	the   tail of the t-distribution from that of the normal distribution (see \eqref{eq:DRellip}).

For numerical illustrations, we consider two specific dispersion matrices, parameterized by $r\in [0,1]$ and $n\in \N$,
	$$
	\Sigma_1=(\sigma_{ij})_{n\times n},~~~\mbox{ where $\sigma_{ii}=1$ and $\sigma_{ij}=r$ for $i\ne j$, and}
	$$
	$$
	\Sigma_2=(\sigma_{ij})_{n\times n},~~~ \mbox{ where  $\sigma_{ii}=1$ and $\sigma_{ij}=r^{|j-i|}$ for $i\ne j$}.
	$$

Note that $\Sigma_1$ represents an equicorrelated model
	and $\Sigma_2$ represents an autoregressive model AR(1).   For $r=0$, $r=1$ or $n=2$, these two dispersion matrices are identical.
	We take four models $\mathbf{X}_i\sim \mathrm N(\boldsymbol{\mu},\Sigma_i)$
	and $\mathbf{Y}_i\sim \mathrm  t(\nu,\boldsymbol{\mu},\Sigma_i)$,  $i=1,2$, and we will let $r,\nu,\alpha,n$ vary.  Note that the location $\boldsymbol{\mu}$ does not matter in computing DQ, and we can simply take $\boldsymbol \mu=\mathbf 0$. The default parameters are set as $r=0.3$, $n=4$, $\nu=3$ and $\alpha =0.05$ if not explained otherwise.

		%


	\subsubsection*{DQ for the t-models  as  the parameter of degrees of freedom $\nu$ varies}
	\begin{figure}[t]
		\caption{DQ and DR based on VaR for $\nu\in (0,10]$  and ES for $\nu\in (1,10]$ with fixed $\alpha =0.05$, $r=0.3$ and $n=4$} 
		\centering\includegraphics[width=15cm]{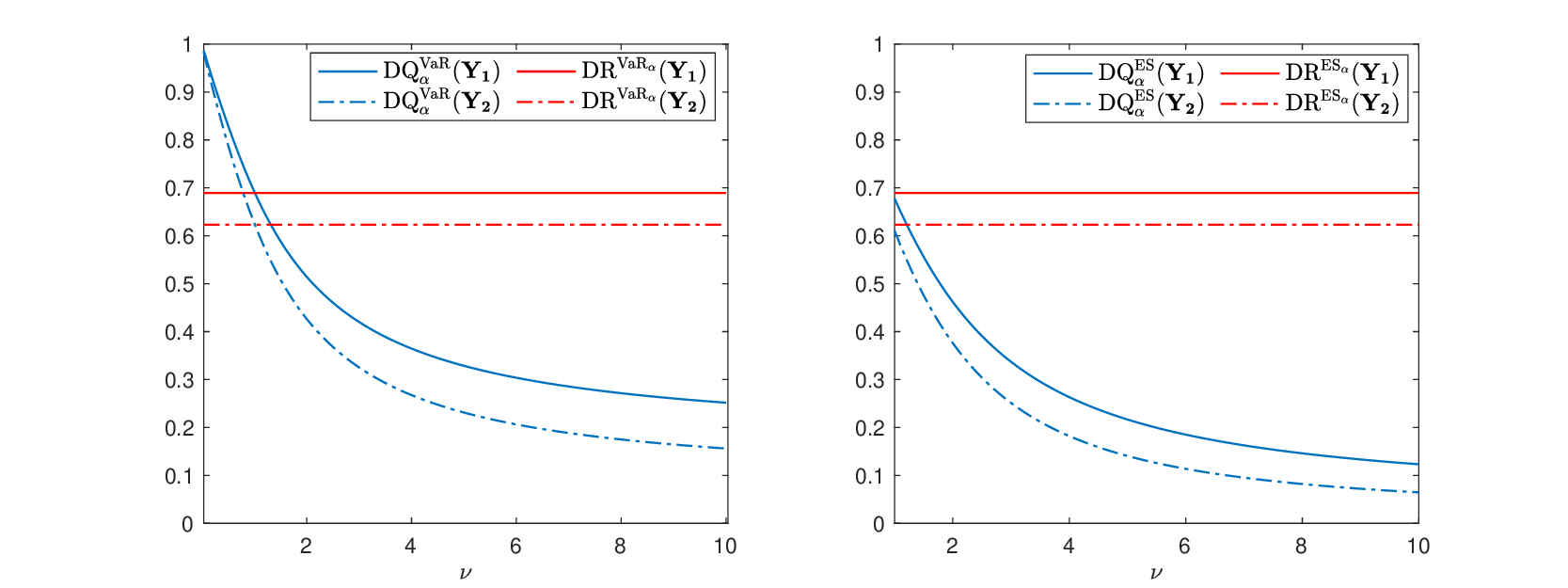}\label{fig:nu_ellip}\end{figure}
	Figure \ref{fig:nu_ellip} presents  the values of DQ for the t-models  with   varying $\nu$, where $\nu\in (0,10]$  for  VaR and $\nu\in (1,10]$ for  ES.
	We observe a monotonic relation that $\mathrm{DQ}^\VaR_\alpha$ and $\mathrm{DQ}^\ES_\alpha$ are decreasing in $\nu$.  In particular, if $\nu$ is close to $0$, we see that $\mathrm{DQ}^\VaR_\alpha \approx 1$ which means there is almost no diversification effect for such super heavy-tailed models. On the other hand, DR completely ignores $\nu$ and always reports the same value. Note that the values of DQ and DR are not directly comparable as they are not on the same scale. 

	\subsubsection*{DQ for elliptical models as  the correlation parameter $r$ varies}
	\begin{figure}[t]
		\caption{DQ based on VaR and ES for $r\in [0,1]$ with fixed $\alpha =0.05$, $\nu=3$,   and $n=4$} \label{fig:r_ellip}\centering\includegraphics[width=15cm]{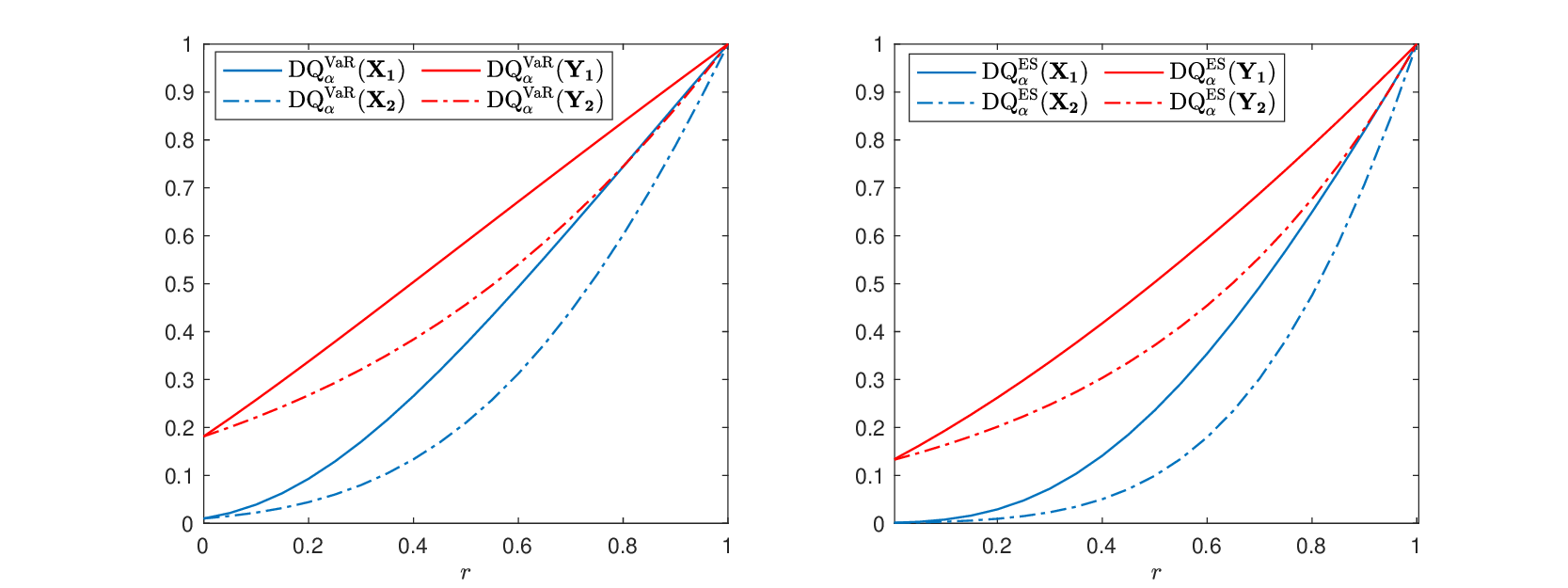}\end{figure}
	In Figure \ref{fig:r_ellip}, we report how DQ changes over $r\in [0,1]$ in the four models. Intuitively, for $r$ close to $1$ which corresponds to comonotonicity,   DQ is close to $1$ in all models since there is no or very weak diversification in this case.
	More interestingly, for $r$ close to $0$, there is very strong diversification for the normal models, meaning $\mathrm{DQ}^\VaR_\alpha \approx 0$
	and
	$\mathrm{DQ}^\ES_\alpha \approx 0$,
	whereas for the t-models,
	$\mathrm{DQ}^\VaR_\alpha $ and $\mathrm{DQ}^\ES_\alpha$ are clearly away from $0$. 
	Note that the components of a t-distribution are tail dependent even for zero or negative correlation (see Example 7.39 of \cite{MFE15}). Hence, DQ is able to capture dependence created by the common factor in the t-model, in addition to its correlation structure. 
	
	\subsubsection*{DQ for varying $\alpha$ and its limit}
	
	\begin{figure}[t]
		\caption{DQ based on VaR and ES for $\alpha \in (0,0.1)$ with fixed $\nu=3$, $r=0.3$ and $n=4$} \label{fig:ellip}\centering\includegraphics[width=15cm]{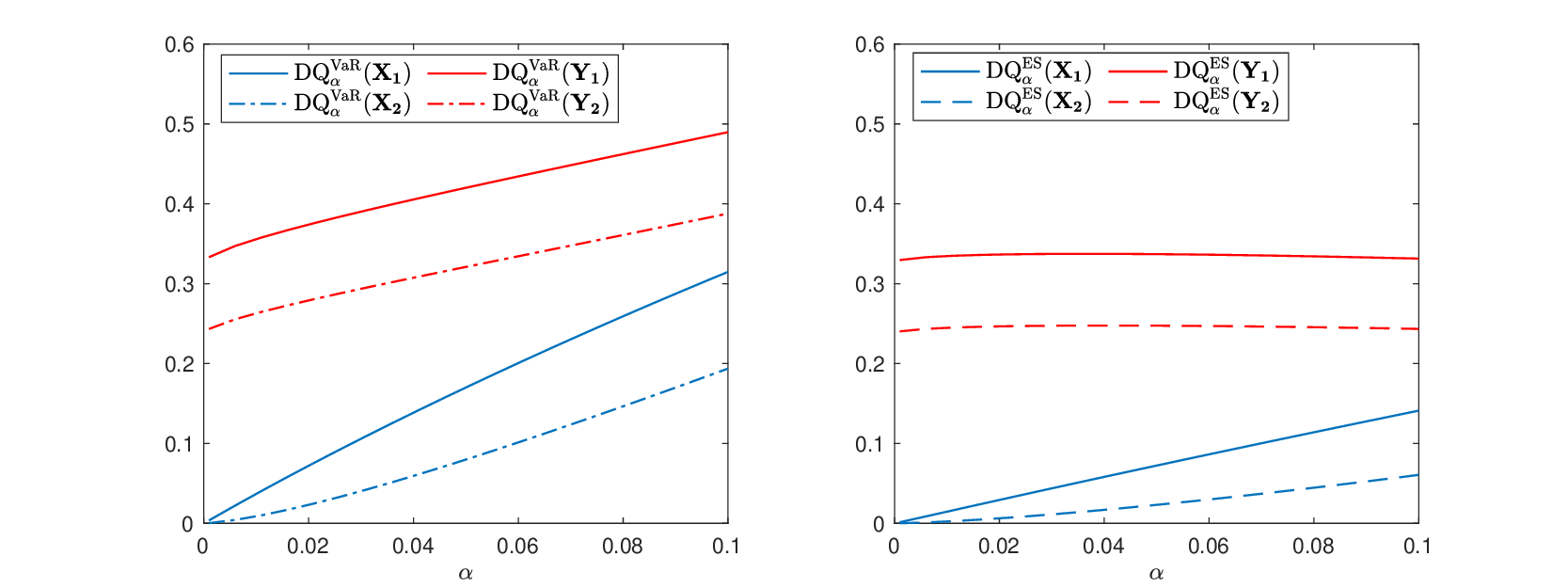}\end{figure}
	
	In Figure \ref{fig:ellip}, we report $\mathrm{DQ}^\VaR_\alpha$ and  $\mathrm{DQ}^\ES_\alpha$ for $\alpha \in(0,1)$ in the four models  with correlation matrices specified  in  Section \ref{sec:normal-t}.
	We can see from Figure \ref{fig:ellip} that DQ  can be non-monotonic with  respect to $\alpha$ (see the curves of ${\rm DQ}^{\ES}_{\alpha}$ for  $\mathbf{X}_i\sim \mathrm t(\nu,\boldsymbol{\mu},\Sigma_i)$). In addition,  we can compute $k_{\Sigma_1}=1.4510$ and $k_{\Sigma_2}=1.6046$. Hence, it can be anticipated from Proposition \ref{prop:comp_Dvar} that, since DQ is decreasing in $k_{\Sigma}$,
	models with $\Sigma_1$ has larger DQ than the corresponding models with  $\Sigma_2$.
	Moreover, as $\alpha\downarrow 0$, we can see that $\mathrm{DQ}^\VaR_\alpha$ converges to its corresponding limits in \eqref{eq:limit-n} and \eqref{eq:limit-t}; also note that  $\mathrm{DQ}^\ES_\alpha$ has the same limits as $\mathrm{DQ}^\VaR_\alpha$ for t-distributions as discussed in Remark \ref{rem:ES-t}.

	\subsubsection*{DQ for elliptical models as  the dimension $n$ varies}
	Figure \ref{fig:n_ellip} is related to Section \ref{sec:normal-t} and  reports how DQ changes over $n\in [2,100]$ in the four models. We choose $r=0.5$ in this experiment for better visibility. As we can see, DQ decreases to $0$ for models with the AR(1) dispersion  $\Sigma_2$,
	and DQ converges to a non-zero constant for models with  the equicorrelated dispersion  $\Sigma_1$.
	This is consistent with Proposition \ref{cor:VaR} (ii) because
	$\mathrm{AC}_{\Sigma_1}\to r$  and $\mathrm{AC}_{\Sigma_2}\to 0$ as $n\to\infty$.

	\begin{figure}[htb!]
		\caption{DQs based on VaR and ES for $n\in [2,100]$ with fixed $\alpha =0.05$, $r=0.5$ and $\nu=3$} \label{fig:n_ellip}\centering\includegraphics[width=15cm]{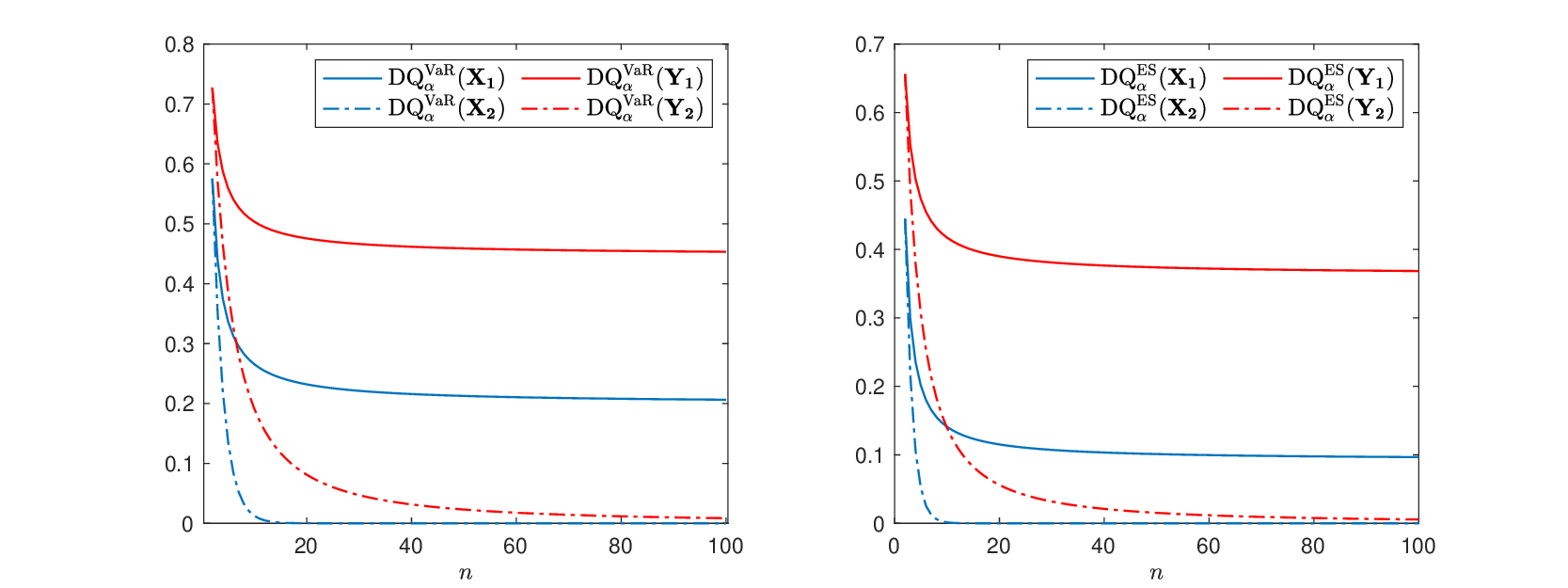}\end{figure}

	\subsubsection*{Cross-comparison between DQ based on VaR and ES}
	
	One may be tempted to compare values of DQ based on VaR to those based on ES.
	Although we see from Figure \ref{fig:ellip} that the curve $\mathrm{DQ}^{\VaR}_\alpha $ often dominates the curve $\mathrm{DQ}^{\ES}_\alpha $ for the same model, such a comparison is not meaningful, since VaR and ES are not meant to be compared at the same  level $\alpha$. For a fair comparison, one needs to associate a VaR level $\alpha$ to an ES level $c\alpha$ where $c\ge 1$ is PELVE of \cite{LW22} defined via  $\ES_{c \alpha}(X) = \VaR_{\alpha}(X)$ for $X$ being normally or t-distributed; note that the location and scale of $X$ do not matter. 
	The values of $c$, $\mathrm{DQ}^{\VaR}_\alpha $  and $\mathrm{DQ}^{\ES}_{c\alpha} $  for $\alpha=0.01$ are summarized in Table \ref{tab:PELVE}.
	As we observe from Table \ref{tab:PELVE}, the values of DQs based on VaR and ES are quite close when the probability level is calibrated via PELVE.
	This is consistent with the afore-mentioned fact that VaR behaves similarly to ES in the setting of elliptical models.
	\begin{table}[H]
		\def\arraystretch{1}
		\begin{center}
			\caption{Values of DQs based on VaR  at level $\alpha=0.01$ and ES    at level $c\alpha$, where $n=4$ and $r=0.3$}
			\label{tab:PELVE}
			\begin{tabular}{c|ccccc}
				&$ c $& $c\alpha$& $\mathrm{DQ}^{\VaR}_\alpha $ &   $\mathrm{DQ}^{\ES}_{c\alpha} $ \\ \hline
				$\mathbf{X}_1\sim \mathrm N(\boldsymbol{\mu},\Sigma_1)$   &2.58    &0.0258& 0.0369& 0.0377 \\ \hline
				$\mathbf{X}_2\sim \mathrm N(\boldsymbol{\mu},\Sigma_2)$   &2.58  & 0.0258& 0.0024&  0.0025 \\ \hline
				$\mathbf{Y}_1\sim \mathrm  t(3,\boldsymbol{\mu},\Sigma_1)$& 3.31 & 0.0331  & 0.3558 &0.3373  \\ \hline
				$\mathbf{Y}_2\sim \mathrm  t(3,\boldsymbol{\mu},\Sigma_2)$ & 3.31 & 0.0331  & 0.2094  & 0.1961
				\\ \hline \hline
			\end{tabular}
		\end{center}
	\end{table}

	\section{Multivariate regularly varying  models}
	\label{sec:MRV}
	 Heavy-tailed  distributions are known to exhibit complicated and even controversial phenomena in finance (see e.g., \cite{IJW11}), 
	 and they are typically modelled via 
	  multivariate regularly varying  (MRV) models, important objects in 
	Extreme Value Theory.
	Such models are particularly relevant for  tail risk measures such as VaR and ES at high levels (\cite{MFE15}). 
	In particular, MRV models  have been applied to DR based on VaR (e.g., \cite{MR10} and \cite{ME13}). 
{Since $\VaR_\alpha(X)/\ES_\alpha (X)\to (\gamma-1)/\gamma$ as $\alpha \downarrow 0$ for $X\in \mathrm{RV}_\gamma$ with finite mean (see e.g., \citet[p.154]{MFE15}), we only present the case of VaR. }

	\begin{definition}
		A random vector $\mathbf X\in\X^n$ has an  MRV model  with some $\gamma >0$ if there exists a Borel probability measure $\Psi$ on the unit sphere $\mathbb {S}^{n}:=\left\{\mathbf{s} \in \mathbb{R}^{n}:\|\mathbf{s}\|=1\right\}$
		such that for any $t>0$ and any Borel set $S \subseteq \mathbb{S}^{n}$ with $\Psi(\partial S)=0$,
		$$
		\lim _{x \rightarrow \infty} \frac{\p (\|\mathbf{X}\|>t x, ~\mathbf{X}/\|\mathbf{X}\| \in S)}{\p (\|\mathbf{X}\|>x)}=t^{-\gamma} \Psi(  S), $$	where $\|\cdot\|$  is the $L_1$-norm (one could use any other norm equivalent to the $L_1$-norm). We call  $\gamma$ the tail index of  $\mathbf X$ and $\Psi$  the spectral measure of $\mathbf X$.  This is written as $\mathbf X \in \mathrm{MRV}_{\gamma}(\Psi)$.
	\end{definition}
       The univariate regular variation   with   tail index $\gamma$ is defined as 
 	$$
 	\mbox{for all~} t>0, ~\lim _{x \rightarrow \infty} \frac{1-F_{X}(t x)}{1-F_{X}(x)}=t^{-\gamma},
 	$$
 where $F$ is the distribution function of $X$.
We write $X\in {\rm R  V}_{\gamma}$ for this property.       
	As a consequence of $\mathbf X \in \mathrm{MRV}_{\gamma}(\Psi)$, $\|\mathbf{X} \|$ satisfies univariate regular variation   with the same  tail index $\gamma$. 

Regular variation is one of the basic notions for describing heavy-tailed distributions and dependence in the tails.  In what follows, we limit our discussion to $\mathbf X \in \mathrm{MRV}_{\gamma}(\Psi)$ under the non-degeneracy condition: 
 $$\Psi\left(\left\{\mathbf s\in \mathbb{S}^n: \mathbf {s}\in (0,\infty)^n\right\}\right) >0.$$
Note that if $\mathbf X \in \mathrm{MRV}_{\gamma}(\Psi)$ satisfies non-degeneracy condition, we have $\mathbf w^\top \mathbf X \in \mathrm{RV}_\gamma$ (See \cite{ME13}).

Let $\mathbf X \in \mathrm{MRV}_{\gamma}(\Psi)$ be  a random vector with identical marginals. If  $ X_1,\dots,X_n$   have a finite mean, then VaR is asymptotically  subadditive in the following sense (see e.g., \cite{ELW09})
	$$
	\VaR_{\alpha} \left(\sum_{i=1}^n X_i\right)   \le { \sum_{i=1}^n \VaR_{\alpha  }(X_i)} \mbox{~~~for $\alpha$ close enough to $0$},
	$$
	but the inequality is reversed if $ X_1,\dots,X_n$ do not have a finite mean. {Next, in contrast to Proposition \ref{pro:ind} and Remark \ref{rem:iid}, we will show that DQ based on VaR can be arbitrarily  close to $n$ even if the individual losses are iid. 	} 
	{
\begin{theorem}\label{thm:MRV}
		Suppose that  $\mathbf{X} \in \mathrm{MRV}_{\gamma}(\Psi)$ and $\mathbf{X}$ has positive joint density on the support of $\mathbf{X}$. Then, 
		\begin{equation}\label{eq:asy_MRV}
		\lim _{\alpha \downarrow 0} {\rm DQ}^{\VaR}_{\alpha}(\mathbf X)={\eta_{\mathbf 1}}{\left(\sum_{i=1}^{n} \eta_{\mathbf{e}_{i}}^{1 / \gamma}\right)^{-\gamma}},\end{equation} where $\eta_{\mathbf{x}}=\int_{\mathbb{S}^{n}}\left(\mathbf{x}^\top \mathbf{s}\right)_+^{\gamma} \Psi(\d \mathbf{s})$ for $\mathbf x\in \R^n$. 
   Moreover, if $X_1, \dots, X_n$ are iid  random variables, then $\mathrm{DQ}_\alpha^{\VaR}(\mathbf X) \to n^{1-\gamma}$ as $\alpha \downarrow 0$.
	\end{theorem}
 \begin{proof}
A more general result of \eqref{eq:asy_MRV} and its proof are  shown in Proposition \ref{prop:lim}, where the asymptotic behavior of $\mathrm{DQ}_\alpha^\VaR$ for weighted portfolios is investigated. 
Since DQ is scale-invariant, by  taking $\mathbf w=(1/n, \dots, 1/n)$ in Proposition \ref{prop:lim},  it gives
 \begin{align*}
 \lim _{\alpha \downarrow 0} {\rm DQ}^{\VaR}_{\alpha}(\mathbf X)=
  \lim _{\alpha \downarrow 0} {\rm DQ}^{\VaR}_{\alpha}(1/n X_1, \dots, 1/n X_n)=\frac{\eta_{\mathbf{w}} }{\left(\sum_{i=1}^{n} w_{i} \eta_{\mathbf{e}_{i}}^{1 / \gamma}\right)^\gamma}, 
 \end{align*}
where $\eta_{\mathbf w}=n^{-\gamma}\int_{\mathbb{S}^n}\left(\mathbf{1}^\top s\right)_+^\gamma \Psi(\d\mathbf{s})=n^{-\gamma}\eta_{\mathbf 1}$. As a result,  we have
$$\lim _{\alpha \downarrow 0} {\rm DQ}^{\VaR}_{\alpha}(\mathbf X)={\eta_{\mathbf 1}}{\left(\sum_{i=1}^{n} \eta_{\mathbf{e}_{i}}^{1 / \gamma}\right)^{-\gamma}}.$$

If $X_1, \dots, X_n$ are iid non-negative random variables, by Example 3.1 of \cite{ELW09}, we have 
 $$\eta_{\mathbf 1}^{1/\gamma}=\lim_{\alpha\downarrow 0}\frac{\VaR_\alpha\left(\sum_{i=1}^n X_i\right)}{\VaR_\alpha(X_1)}=n^{1/\gamma},$$ which implies that $\eta_{\mathbf 1}=n$. 
 Moreover, $$(\eta_{\mathbf e_i})^{1/\gamma}=\lim_{\alpha \downarrow 0}\frac{\VaR_{\alpha} (X_i)}{\VaR_{\alpha}(X_1)}=1.$$ Hence, $\lim _{\alpha \downarrow 0} {\rm DQ}^{\VaR}_{\alpha}(\mathbf X)=n^{1-\gamma}$. Further, if $\gamma\downarrow 0,$ then $ {\rm DQ}^{\VaR}_{\alpha}(\mathbf X)\to n$.
 \end{proof}
	}

The $\alpha$-CE model in Theorem \ref{th:var-01n} with ${\rm DQ}_\alpha^{\VaR} (\mathbf X)= n$  is complicated and involves both positive and negative dependence.
	Theorem \ref{thm:MRV} suggests that 
	${\rm DQ}_\alpha^{\VaR} (\mathbf X)\approx n$ can be obtained for some very heavy-tailed iid model with $\gamma$ close to $0$.
	Therefore, the upper bound $n$ on  ${\rm DQ}_\alpha^{\VaR}$ is  relevant when analyzing very heavy-tailed risks such as catastrophe losses; we refer to \cite{EKM97} for a general treatment of heavy-tailed risks in insurance and finance.
 
{
\begin{remark}
Suppose that  $X_1,\dots,X_n$ are iid random variables with  $X_1 \in \mathrm{RV}_{\gamma}$  having  positive density over its support. We have $\mathbf X=(X_1,\dots, X_n) \in \mathrm{MRV}_{\gamma}(\Psi)$ by \citet[Example 2.1.4]{KS20}, and thus ${\rm DQ}_\alpha^{\VaR} (\mathbf X)\to n^{1-\gamma}$ as $\alpha \downarrow 0$.
\end{remark}
}

{\begin{remark}
We note that the intersection between elliptical  distributions and MRV distributions is non-empty. For  $\mathbf X \sim \mathrm{E}_n(\mathbf \mu, \Sigma, \tau)$, we have  
$$\mathbf X \laweq \mu+RAU,$$
where $A \in \R^{n \times n}$ satisfying $AA^\top=\Sigma$, $U$ is uniformly distributed on the Euclidean sphere $\mathbb S_2^d$ and $R$ is a non-negative random variable that is independent of $U$.  Theorem 4.3 of 
 \cite{HL02} showed that $\mathbf X$ has   an MRV model if and only if $R \in \mathrm{RV}_{\gamma}$ for some $\gamma >0$. Assume that the elliptically distributed 
$\mathbf X$ is in  $\mathrm{MRV}_\gamma(\Psi)$  with $\gamma>0$. As a result, we have $Y \sim E_1(0,1,\tau) \in \mathrm{RV}_{\gamma}$. Let $f$ be the density of $Y$. Following  Proposition \ref{cor:VaR} (i) and the fact that  $\VaR_\alpha(Y)/\ES_\alpha (Y)\to (\gamma-1)/\gamma$ as $\alpha \downarrow 0$ for $Y\in \mathrm{RV}_\gamma$ with finite mean, we have 
$$\lim_{\alpha \downarrow 0} {\rm DQ}^{\ES}_{\alpha}(\mathbf X)=\lim_{\alpha \downarrow 0} {\rm DQ}^{\VaR}_{\alpha}(\mathbf X)=
				\lim_{x \to \infty}k_\Sigma \frac{f(k_\Sigma x)}{f(x)}=k_\Sigma^{-\gamma}.$$
If $\mathbf X$ follows an  elliptical distribution  in the MRV class, then ${\rm DQ}_\alpha^\ES(\mathbf X)$ has the same limit as ${\rm DQ}^{\VaR}_{\alpha}(\mathbf X)$.
For example, if  $\mathbf X \sim t(\nu, \mathbf \mu, \Sigma)$, we have $\mathbf X\in \mathrm{MRV}_{\gamma}(\Psi)$ with $\gamma =\nu$ as we have shown in \eqref{eq:limit-t} that $\lim_{\alpha \downarrow 0} {\rm DQ}^{\VaR}_{\alpha}(\mathbf X)=\lim_{\alpha \downarrow 0} {\rm DQ}^{\ES}_{\alpha}(\mathbf X)=k_\Sigma^{-\nu}$.
 \end{remark}

To end this section, we show that if there exists an asset with a strictly heavier tail than the other assets in the portfolio, then   DQ based on VaR tends to 1 as $\alpha \downarrow 0$.
\begin{proposition}\label{prop:RV}
Suppose  $X_i \in \mathrm{RV}_{\gamma_i}$ for $i\in [n]$  such that $\gamma_1<\min_{i=2, \dots, n}\gamma_i$. If $X_1, \dots, X_n$ have positive densities on their support,  then
$\lim_{\alpha \downarrow 0} \mathrm{DQ}_\alpha^\VaR(\mathbf X)=1.$
\end{proposition}
\begin{proof}
Since $\gamma_1<\min_{i=2, \dots, n}\gamma_i$,  $X_1$ has a heavier tail than $X_2, \dots, X_n$. As a result,  we have $\sum_{i=1}^n X_i \in \mathrm{RV}_{\gamma_1}$  regardless of the dependence between all random variables (See \citet[Lemma 1.3.2]{KS20}), that is, 
$$\lim_{x \to \infty} \frac{\p\left(\sum_{i=1}^n X_i>x\right)}{\p(X_1>x)}=1.$$
Moreover, $X_1$ having a heavier tail than $X_2, \dots, X_n$ also implies that 
$\lim_{\alpha \downarrow 0} {\VaR_{\alpha}(X_i)}/{\VaR_{\alpha}(X_1)}=0$ for all $i=2, \dots, n$, and thus $\lim_{\alpha \downarrow 0} \sum_{i=1}^n \VaR_{\alpha}(X_i)/\VaR_{\alpha}(X_1)=1$.
Therefore, we have 
\begin{align*}
\lim_{\alpha \downarrow 0}
\mathrm{DQ}_\alpha^\VaR(\mathbf X)&=\lim_{\alpha \downarrow 0} \frac{\p\left(\sum_{i=1}^n X_i>\sum_{i=1}^n \VaR_{\alpha}(X_i)\right)}{\alpha}\\
&=\lim_{\alpha \downarrow 0} \frac{\p\left(\sum_{i=1}^n X_i>\sum_{i=1}^n \VaR_{\alpha}(X_i)\right)}{\p(X_1>\VaR_\alpha(X_1))}\\
&=\lim_{\alpha \downarrow 0} \frac{\p\left(\sum_{i=1}^n X_i>\sum_{i=1}^n \VaR_{\alpha}(X_i)\right)}{\p(X_1>\sum_{i=1}^n \VaR_\alpha(X_i))}\frac{\p(X_1>\sum_{i=1}^n \VaR_\alpha(X_i))}{\p(X_1>\VaR_\alpha(X_1))}\\
&=\lim_{\alpha \downarrow 0}\left(\frac{\sum_{i=1}^n\VaR_{\alpha}(X_i)}{\VaR_\alpha(X_1)}\right)^{-\gamma_1}=1.
\end{align*}
Thus, we get the desired result.
\end{proof}
Proposition \ref{prop:RV}  illustrates the intuitive fact that, if the tail of one asset  is  strictly heavier than the   others, then the portfolio has no diversification in the tail region, i.e., as $\alpha \downarrow 0$.

	\section{Optimization for the elliptical  models and MRV models}\label{sec:5}
	 We analyze portfolio diversification for a 
	random vector  $\mathbf{X}\in\X^n$ representing losses from  $n$ assets and a vector $\mathbf w= (w_{1},  \dots, w_{n}) \in \Delta_{n}$ of portfolio weights, where  $$  \Delta_{n}:=\left\{\mathbf{x} \in[0,1]^{n}: x_{1} +\dots+x_{n}=1\right\}.$$ 
	The total loss  of the portfolio is  $\mathbf{w}^\top \mathbf{X}$. 
	We write $\mathbf w \odot \mathbf X=\left(w_1X_1,\dots,w_nX_n\right)$ which is the portfolio loss vector with the weight $\mathbf w$.
	For a portfolio selection problem, we need to treat  
	$\mathrm{DQ}^\rho_\alpha(\mathbf w \odot \mathbf X)$ as a function of the portfolio weight $\mathbf w$.   

	\cite{HLW22} studied  the following optimization diversification problem
	\begin{equation}\label{eq:optimal_DQ}
		\min _{\mathbf{w} \in \Delta_n}  {\rm DQ}^{\VaR}_{\alpha}(\mathbf w \odot \mathbf X) \mbox{~~~and~~~} \min _{\mathbf{w} \in \Delta_n}  {\rm DQ}^{\ES}_{\alpha}(\mathbf w \odot \mathbf X);
	\end{equation}
	for general $\mathbf X$.  
	 Moreover,   efficient algorithms are obtained to  optimize  ${\rm DQ}^{\VaR}_{\alpha}$ and ${\rm DQ}^{\ES}_{\alpha}$ in real-data applications; see their Sections 6.2 and 7.  
In this section, we focus on the portfolio optimization problems for  elliptical and MRV models.

 For the elliptical  models, 	the optimization of $\mathrm{DQ}^\VaR_\alpha$, $\mathrm{DQ}^\ES_\alpha$   boils down to maximizing $k_{\mathbf w \Sigma \mathbf w^\top}$ in \eqref{eq:k} since DQ of $\mathbf w\odot \mathbf X$ is decreasing in $k_{\mathbf w \Sigma \mathbf w^\top}$. 
	We assume that $\Sigma$ is invertible, and write
	$\Sigma=(\sigma_{ij})_{n\times n},$  with diagonal entries $\sigma_{ii}=\sigma_i^2$, $i\in [n]$, and $\boldsymbol \sigma=(\sigma_{1},\dots,\sigma_{n})$.
	Note that
	$$
	k_{\mathbf w \Sigma \mathbf w^\top} =\frac{\mathbf w^\top \boldsymbol{\sigma}}{ \sqrt{\mathbf w^\top\Sigma \mathbf w}},
	$$
	and we immediately give the  optimizer of \eqref{eq:optimal_DQ}  for the elliptical  models.

 \begin{theorem}\label{th:opt_elli}
	Suppose that  $\mathbf X \sim  \mathrm{E}_{n}(\boldsymbol{\mu}, \Sigma, \tau)$, 
	$\Sigma$ is invertible and $\alpha \in (0,1/2)$, then the vector
	\begin{equation}\label{eq:opt-w}
		\mathbf w^*=\argmax_{\mathbf w\in \Delta_n} \frac{\mathbf w^\top \boldsymbol{\sigma}}{ \sqrt{\mathbf w^\top\Sigma \mathbf w}}
	\end{equation}
	minimizes \eqref{eq:optimal_DQ}, that is, \begin{equation}\label{eq:opt}
		\min_{\mathbf w \in \Delta_n} {\rm DQ}^{\rho}_{\alpha}(\mathbf w \odot \mathbf X)={\rm DQ}^{\rho}_{\alpha}(\mathbf w^* \odot \mathbf X)
	\end{equation} for  $\rho$ being  $\VaR$ or  $\ES$.
	\end{theorem}
	
	The optimization problem  \eqref{eq:opt-w} is well studied in the literature, and the existence   and  uniqueness   of the solution can be verified if ${\Sigma}$ is invertible, see, e.g.~\cite{CC08}. Note that the optimizer for problem \eqref{eq:opt} does not depend on the tail  probability level $\alpha$.
	It is straightforward to see that
	$$ \argmin_{\mathbf w \in \Delta_n} {\rm DR}^{\rho_\alpha}(\mathbf w \odot \mathbf X)=\argmax_{\mathbf w\in \Delta_n} \frac{\mathbf w^\top\boldsymbol{\mu}+ \mathbf w^\top \boldsymbol{\sigma}\rho_{\alpha} (Y)}{\mathbf w^\top\boldsymbol{\mu}+\sqrt{\mathbf w^\top\Sigma \mathbf w}\rho_{\alpha} (Y)}$$ for $\rho$ being  $\VaR$ or  $\ES$ and $Y\sim \mathrm{E}_1(0,1,\tau)$. This optimizer is the same as that of  \eqref{eq:opt} if  $\boldsymbol{\mu}=\mathbf 0$.
	This shows that for centered elliptical models, optimizing DQ and optimizing DR are equivalent problems, both of which are further  equivalent to optimizing DR based on SD (assuming it exists).  This is intuitive as for a fixed $\tau$, centered elliptical distributions are parameterized by their dispersion matrices. 
	
	\begin{example} 	Assume  that $\mathbf{X}\sim \mathrm  t(\nu,\boldsymbol{\mu},\Sigma)$ where $\nu=3$ and  the dispersion matrix is  given by  $$\Sigma=\left(\begin{array}{cc}1 & 0.5 \\ 0.5 & 2\end{array}\right).$$
		Clearly, DQ does not depend on $\boldsymbol \mu$.  We show the curves of    ${\rm DQ}^{\VaR}_{\alpha}(\mathbf w \odot \mathbf X)$
		and   ${\rm DQ}^{\ES}_{\alpha}(\mathbf w \odot \mathbf X)$
		against the weight $w_1$ with various values of  $\alpha=0.001, 0.01, 0.025, 0.05$.  It can be anticipated  from \eqref{eq:opt-w} that although DQ depends on $\alpha$, the optimizer does not.  By solving \eqref{eq:opt-w},  we  get $w^*_1=0.5860$ and $w_2^*=0.4140$, which corresponds to the observations in Figure \ref{fig:ellip2}. Recall  that ${\rm DQ}^{\ES}_{\alpha}$ is quite flat when $\alpha$ varies in Figure \ref{fig:ellip}, and hence curves of ${\rm DQ}^{\ES}_{\alpha}(\mathbf w \odot \mathbf X)$ look similar for different $\alpha$.
		\begin{figure}[htb!]
			\caption{Values of ${\rm DQ}^{\VaR}_{\alpha}(\mathbf w \odot \mathbf X)$ and ${\rm DQ}^{\ES}_{\alpha}(\mathbf w \odot \mathbf X)$ for $w_1\in[0,1]$}\label{fig:ellip2}
			\centering
			\includegraphics[width=14cm]{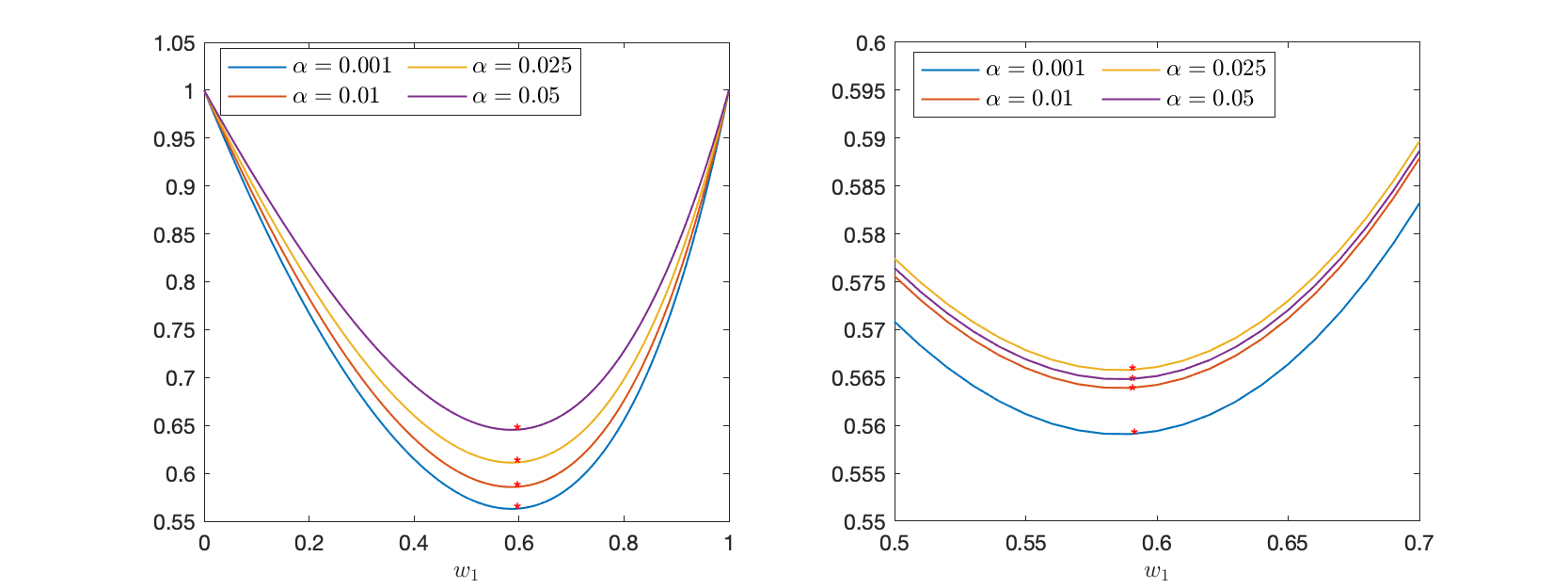}
		\end{figure}
	\end{example}

	%

	Next, we turn to the MRV model. The following result gives the limit of DQ of the portfolio
	$\mathbf w \odot \mathbf X$
	where $\mathbf X$ follows an MRV model. {Due to the same reason stated in Section \ref{sec:MRV}, 
 we only present the case of VaR.}
 In the proofs below, for any positive functions $f$ and $g$, we write
	$
	f(x) \simeq g(x)  \text { as } x \rightarrow x_{0}
	$
	to represent  
	$
	\lim _{x \rightarrow x_{0}}  {f(x)}/{g(x)}=1
	$. 
	
	\begin{proposition}\label{prop:lim}
		Suppose that  $\mathbf{X} \in \mathrm{MRV}_{\gamma}(\Psi)$ and $\mathbf{X}$ has positive joint density on the support of $\mathbf{X}$. Then, for $\mathbf w\in \Delta_n$,
		$$
		\lim _{\alpha \downarrow 0} {\rm DQ}^{\VaR}_{\alpha}(\mathbf w \odot \mathbf X)=f(\mathbf w),$$ where  $f(\mathbf{w})=\eta_{\mathbf{w}}/\left(\sum_{i=1}^{n} w_{i} \eta_{\mathbf{e}_{i}}^{1 / \gamma}\right)^\gamma
		$ and $\eta_{\mathbf{x}}=\int_{\mathbb{S}^{n}}\left(\mathbf{x}^\top \mathbf{s}\right)_+^{\gamma} \Psi(\d \mathbf{s})$ for $\mathbf x\in \R^n$.
	\end{proposition}
	\begin{proof}
		If  $\mathbf{X} \in \mathrm{MRV}_{\gamma}(\Psi)$ with $\gamma \in (0,1)$, we have (Lemma  2.2 of \cite{ME13})
		$$
		\lim _{\alpha \downarrow 0}\frac{ \VaR_{\alpha} \left(\sum_{i=1}^n w_iX_i\right) }{ \VaR_{\alpha} \left(\|\mathbf{X}\|_1\right)} =  \eta_{\mathbf{w}}^{1 / \gamma},
		$$
		and
		$$
		\lim _{\alpha \downarrow 0}\sum_{i=1}^n \frac{w_i \VaR_{\alpha} \left(X_i\right) }{ \VaR_{\alpha} \left(\|\mathbf{X}\|_1\right)} = \sum_{i=1}^n w_i\eta_{\mathbf{e}_i}^{1 / \gamma},
		$$
		where $\|\mathbf{X}\|_1=\sum_{i=1}^n |X_i|$.
		As $\mathbf{X}$ has positive joint density, $\VaR_\alpha$ is continuous for $\sum_{i=1}^n w_i X_i$. Then we have $\VaR_{\alpha^*}(\sum_{i=1}^n w_i X_i)=\sum_{i=1}^n w_i\VaR_{\alpha} (X_i)$. Thus, it follows that
		$$
		\frac{ \VaR_{\alpha} \left(\sum_{i=1}^n w_iX_i\right) }{\VaR_{\alpha^*}(\sum_{i=1}^n w_iX_i)} \to  \frac{\eta_{\mathbf{w}}^{1 / \gamma}}{\sum_{i=1}^{n} w_{i} \eta_{\mathbf{e}_{i}}^{1 / \gamma}}
		\mbox{~~~~as } \alpha \downarrow 0.
		$$ Since $\sum_{i=1}^nw_i X_i\in {\rm RV}_\gamma$, for $c>0$,
		$$\frac{ \VaR_{\alpha} \left(\sum_{i=1}^nw_i X_i\right) }{\VaR_{c\alpha} \left(\sum_{i=1}^n w_i X_i\right) }\simeq \left(\frac{1 }{c
		}\right)^{-1/\gamma}
		\mbox{~~~~as }  \alpha \downarrow 0.$$
		Let $c=\alpha^*/\alpha$, we have
		$$ \left(\frac{\alpha}{\alpha^*
		}\right)^{-1/\gamma} \to \frac{\eta_{\mathbf{w}}^{1 / \gamma}}{\sum_{i=1}^{n} w_{i} \eta_{\mathbf{e}_{i}}^{1 / \gamma}}.$$
		Hence,
		$$
		{\rm DQ}^{\VaR}_{\alpha}(\mathbf w \odot \mathbf X) = \frac{\alpha^*}{\alpha}\to \frac{\eta_{\mathbf{w}}}{\left(\sum_{i=1}^{n} w_{i} \eta_{\mathbf{e}_{i}}^{1 / \gamma}\right)^\gamma}.
		$$
		The desired result is obtained.
	\end{proof}

	Proposition \ref{prop:lim} allows us to approximately optimize $\mathrm{DQ}_\alpha^\VaR$ by minimizing $f(\mathbf w)$.  For $\mathbf{X} \in \operatorname{MRV}_{\gamma}(\Psi)$ with $\gamma >1$, by assuming  $\Psi\left(\left\{\mathbf{x} \in \R^n: \mathbf{a}^\top \mathbf{x}=0\right\}\right)=0$ for any $\mathbf{a} \in \mathbb{R}^{n}$, which means that all components  are relevant for the extremes of $\mathbf X$,   the existence  and uniqueness of  $\mathbf{w}^{*}=\argmin _{\mathbf{w} \in \Delta_{n}} f(\mathbf w)$   are  guaranteed.  In fact, the existence of $\mathbf w^{*}$ is due to the continuity of $f(\mathbf w)$ and the compactness of $\Delta_{n}$. To show   uniqueness, we can rewrite  the above minimization problem  as
	$$
	\min _{\mathbf{w}\in\Delta_n} \eta_{\mathbf{w}}
	~~~~\text{s.t.}~~~~\sum_{i=1}^{d} w_{i} \eta_{\mathbf{e}_{i}}^{1 / \gamma}=1.$$  Note that   the set of constraints is compact and  $\eta_{\mathbf{w}}$ is strictly convex, and hence $\mathbf{w}^{*}$ is unique.
	
	\begin{example} Assume that  $Y_{1}$ and $Y_{2}$ are iid following a standard t-distribution with  $\nu>1$ degrees of freedom. A random vector $\mathbf{X}=\left(X_{1}, X_{2}\right)$ is defined as
		$$
		\mathbf{X}=A \mathbf{Y} \quad \text{with}~~ A=\left(\begin{array}{cc}
			1 & 0 \\
			r & \sqrt{1-r^{2}}
		\end{array}\right).
		$$
		The random vectors $\mathbf X$ and $\mathbf Y$ are not elliptically distributed.
		Using the results in \cite{ME13},  we have
		$$
		\frac{\eta_{\mathbf{w}}}{\eta_{\mathbf{1}_{1}}}=\left(w_{1}+w_{2} r\right)^{\nu}+\left(w_{2} \sqrt{1-r^{2}}\right)^{\nu},
		$$
		and
		$$
		\frac{\eta_{\mathbf{w}}}{\eta_{\mathbf{1}_{2}}}=\frac{\left(w_{1}+w_{2} r\right)^{\nu}+\left(w_{2} \sqrt{1-r^{2}}\right)^{\nu}}{r^{\nu}+\sqrt{1-r^{2}}^{\nu}}.
		$$
		Hence,  $$
		f(\mathbf w)=\!\left(w_{1}\left(\left(w_{1}+w_{2} r\right)^{\nu}\!+\left(w_{2} \sqrt{1\!-r^{2}}\right)^{\nu}\right)^{-\frac{1}{\nu}}\!+w_{2}\left(\frac{\left(w_{1}\!+w_{2} r\right)^{\nu}\!+\left(w_{2} \sqrt{1\!-r^{2}}\right)^{\nu}}{r^{\nu}\!+\sqrt{1-r^{2}}^\nu}\right)^{\!-\frac{1}{\nu}}\right)^{-\nu}.
		$$
		Take  $r=0.3$.
		We show the curves of  ${\rm DQ}^{\VaR}_{\alpha}(\mathbf w \odot \mathbf X)$ against  $w_{1}$ for   $\alpha=0.001, 0.01, 0.025$ and  $\nu=2, 4$. Also, we use $f(\mathbf w)$ to approximate  ${\rm DQ}^{\VaR}_{\alpha}(\mathbf w \odot \mathbf X)$ as $\alpha$ tends to 0. From Figure \ref{fig:MRV}, we can see that  the optimizer  $w^*_1$ is converging to the one that maximizes $f(\mathbf w)$ as $\alpha$ tends to $0$.
		\begin{figure}[htb!]
			\caption{Values of ${\rm DQ}^{\VaR}_{\alpha}(\mathbf w \odot \mathbf X)$ with $\nu=2$ (left) and $\nu=4$ (right)}\label{fig:MRV}
			\includegraphics[width=15cm]{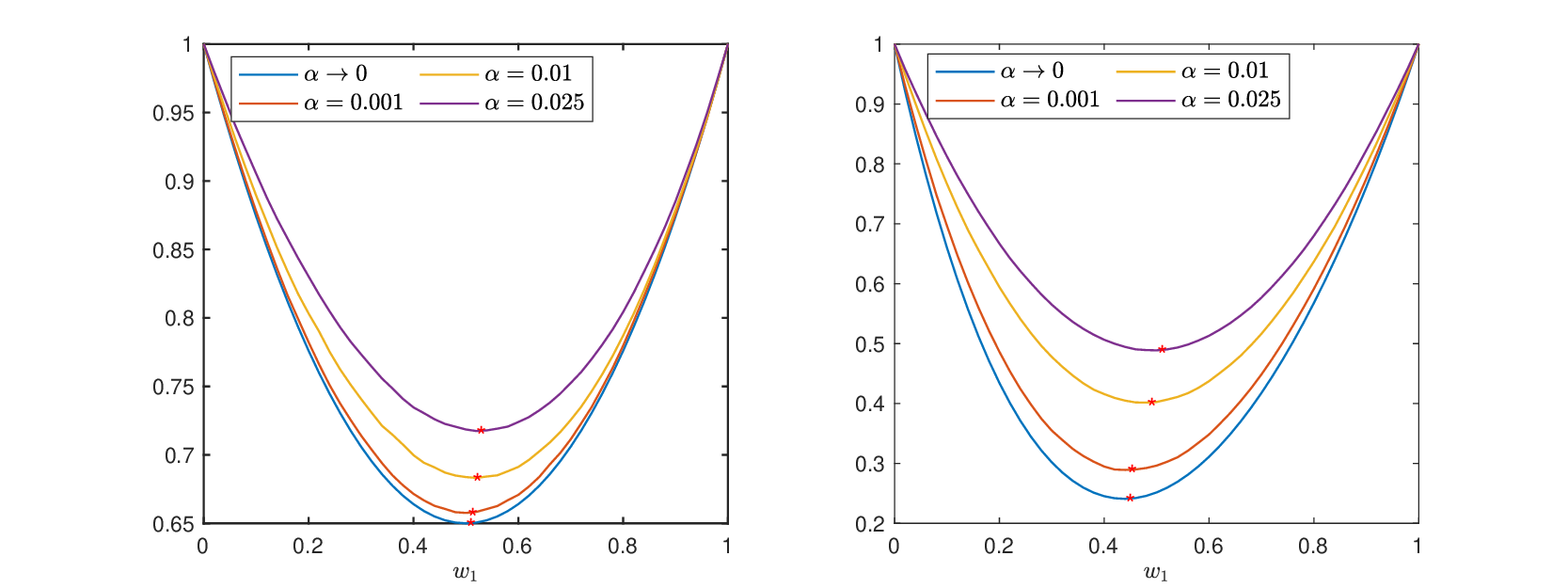}
		\end{figure}
	\end{example}

	{
 \begin{remark} 
 Some negative dependence concepts yield small values of DQ. The joint mix dependence usually leads to a zero DQ as we see in Theorem \ref{th:var-01n} (ii). The negative dependence concept of \cite{LA14}, weaker than joint mix dependence,  does not necessarily lead to a small value of DQ. For instance,  the portfolio vector $\textbf{X}=(X,-\epsilon X)$
 is counter-monotonic for $\epsilon>0$, but its DQ can be close to $1$ for small $\epsilon$.
 In particular,  we have $\mathrm{DQ}_\alpha^{\VaR}(\textbf X)\approx0.9333$ and $\mathrm{DQ}_\alpha^{\ES}(\textbf X)\approx0.9044$ for $\alpha=0.05$ and $\epsilon=0.01$ when $X$ follows a standard normal distribution.   
\end{remark}
 }

	\section{Conclusion}\label{sec:6}
	
	The DQs based on VaR and ES are investigated in this paper, following the theory of DQ in \cite{HLW22}.  In particular,  for  elliptical and MRV  models, these DQs have simple forms. Comparisons between DQ and  DR illustrate  some attractive features of DQ. These results  enhance the theory and applications of DQ.
	
	We summarize some features below. (i) In cases of VaR and ES, DQs have simple formulas, in a way comparable to DRs.  (ii) DQs based on VaR and ES take values in  bounded intervals and have  natural ranges of  $[0, n]$  and  $[0, 1]$, respectively.  The special values $0$, $1$ and $n$ which correspond to   special dependence structures can be constructed. (iii)  DQs based on VaR and ES for elliptical distributions and MRV models  have  convenient expressions   and it can capture heavy tails in an intuitive way.   (iv)   Portfolio optimization  for elliptical models   boils down to a well-studied problem in the literature. For  centered elliptical models, optimizing DQ and optimizing DR are equivalent problems.  
 
	We discuss some future directions for the research of DQ. As a potential alternative to ES, expectiles (\cite{BKMR14}) have received increasing attention in the recent literature; indeed, they are the only elicitable coherent risk measures (\cite{Z16}). It would be  interesting to formulate DQ based on expectiles and  investigate its  properties that are different from DQ based on ES or VaR. As another interesting class of risk measures, the optimized certainty equivalents (\cite{BT07}) are introduced from decision-theoretic criterion based
on utility functions. It would be useful to construct DQ based on utility functions or optimized certainty equivalents and analyze the decision-theoretic implications of DQ.

\vspace{7mm}	
	\textbf{Acknowledgements.}  The authors  thank the editor and two anonymous referees for constructive comments. The authors also thank Fabio Bellini,  Ale\v s \v Cern\'y, An Chen, Jan Dhaene, Paul Embrechts, Paul Glasserman, 
	Pablo Koch-Medina, Dimitrios Konstantinides,  Henry Lam, Jean-Gabriel Lauzier, Fabio Macherroni, Massimo Marinacci, Cosimo Munari, Alexander Schied,  Qihe Tang, Stan Uryasev, Peter Wakker, and Zuoquan Xu for helpful comments and discussions on various issues of DQ. 
Xia Han is supported by the Fundamental Research Funds for the Central Universities, Nankai University (Grant No. 63231138).
Liyuan Lin is supported by the Hickman Scholarship from the Society of Actuaries. 
	Ruodu Wang is supported by the Natural Sciences and Engineering Research Council of Canada (RGPIN-2018-03823, RGPAS-2018-522590) and the Sun Life Research Fellowship.



	\setcounter{table}{0}
	\setcounter{figure}{0}
	\setcounter{equation}{0}
	\renewcommand{\thetable}{EC.\arabic{table}}
	\renewcommand{\thefigure}{EC.\arabic{figure}}
	\renewcommand{\theequation}{EC.\arabic{equation}}
	
	\setcounter{theorem}{0}
	\setcounter{proposition}{0}
	\renewcommand{\thetheorem}{EC.\arabic{theorem}}
	\renewcommand{\theproposition}{EC.\arabic{proposition}}
	\setcounter{lemma}{0}
	\renewcommand{\thelemma}{EC.\arabic{lemma}}
	
	\setcounter{corollary}{0}
	\renewcommand{\thecorollary}{EC.\arabic{corollary}}
	
	\setcounter{remark}{0}
	\renewcommand{\theremark}{EC.\arabic{remark}}
	\setcounter{definition}{0}
	\renewcommand{\thedefinition}{EC.\arabic{definition}}



\begin{thebibliography}{99}
		\small

{
\bibitem[\protect\citeauthoryear{Aleksandr and Khinchin}{Aleksandr and Khinchin}{1949}]{AK49}
{Aleksandr, I.  and Khinchin, A.} (1949). \emph{Mathematical Foundations of Statistical Mechanics}. Courier Corporation.}
  
\bibitem[\protect\citeauthoryear{Artzner et al.}{Artzner et al.}{1999}]{ADEH99} {Artzner, P., Delbaen, F., Eber, J.-M. and Heath, D.} (1999). Coherent measures of risk. \emph{Mathematical Finance}, \textbf{9}(3), 203--228.
		
		\bibitem[\protect\citeauthoryear{{Basel Committee on Banking
				Supervision}}{{BCBS}}{2019}]{BASEL19}
		{BCBS} (2019).
		{\em  Minimum Capital Requirements for Market Risk.  February 2019.}
		Basel Committee on Banking
		Supervision. Basel: Bank for International Settlements, Document d457.
		
		
		
		
		\bibitem[\protect\citeauthoryear{Bellini et al.}{2014}]{BKMR14}
		Bellini, F., Klar, B., M\"uller, A. and Rosazza Gianin, E. (2014). Generalized quantiles as risk measures.\emph{ Insurance: Mathematics and Economics}, \textbf{54}(1), 41--48.
		
		
		\bibitem[\protect\citeauthoryear{Ben-Tal and Teboulle}{2007}]{BT07}Ben-Tal, A. and Teboulle, M. (2007). An old-new concept of convex risk measures: The optimized certainty equivalent. \emph{Mathematical Finance}, \textbf{17}(3), 449--476.
		
		
		\bibitem[\protect\citeauthoryear{Bignozzi et al.}{2016}]{BMWW16}
		Bignozzi, V., Mao, T., Wang, B. and Wang, R. (2016). Diversification limit of quantiles under dependence
		uncertainty. \emph{Extremes}, \textbf{19}(2), 143--170.
		%
		
		
		
		
		
		
		\bibitem[\protect\citeauthoryear{B\"urgi et al.}{2008}]{BDI08} B\"urgi, R.,  Dacorogna, M. M. and  Iles,  R. (2008). Risk aggregation, dependence structure and diversification benefit.  \emph{In: Stress Testing for Financial Institutions}. Riskbooks, Incisive
Media, London,  265--306.
		
		
		
		%
		
		
		
		
		
		
		
		
		
		
		
		\bibitem[\protect\citeauthoryear{Choueifaty and Coignard}{2008}]{CC08}
		Choueifaty, Y. and Coignard, Y. (2008). Toward maximum diversification. \emph{The Journal of Portfolio Management}, \textbf{35}(1), 40--51.
		
		
		
		\bibitem[\protect\citeauthoryear{Cui et al.}{2022}]{CTYZ22} Cui, H., Tan, K. S., Yang, F.  and Zhou, C. (2022) Asymptotic analysis of portfolio diversification. \emph{Insurance: Mathematics and Economics}, \textbf{106}, 302--325.
		
		
		
		
		
		
		
		
		
		
		
		
		
		\bibitem[\protect\citeauthoryear{Delbaen}{Delbaen}{2002}]{D02}
		Delbaen, F. (2002). Coherent risk measures on general probability spaces. \emph{Advances in Finance and Stochastics: Essays in Honour of Dieter Sondermann}, 1--37. Eds: K. Sandmann, P. J. Sch\"onbucher.
		
		
		\bibitem[\protect\citeauthoryear{Dhaene et al.}{Dhaene et al.}{1999}]{DD99}
		{Dhaene, J. and Denuit, M.} (1999). {The safest dependence structure among risks}. \emph{Insurance: Mathematics and Economics},
		\textbf{25}(1), 11--21.
		
	{\bibitem[\protect\citeauthoryear{Durrett}{Durrett}{2019}]{D19}	
		Durrett, R. (2019). \emph{Probability: Theory and Examples}. Vol. 49. Cambridge University Press.}
		
		
		
		
		
		
		\bibitem[\protect\citeauthoryear{EIOPA}{EIOPA}{2011}]{E11}
		EIOPA (2011). Equivalence assessment of the Swiss supervisory system in relation to articles 172, 227
		and 260 of the Solvency II Directive, {EIOPA-BoS-11-028}. 
		
		
		\bibitem[\protect\citeauthoryear{Embrechts et al.}{1997}]{EKM97}
		Embrechts,  P., Kl\"uppelberg, C. and Mikosch, T. (1997). \emph{Modelling Extremal Events for Insurance and Finance.} Springer, Berlin.
		
		\bibitem[\protect\citeauthoryear{Embrechts et al.}{2009}]{ELW09}
		Embrechts, P., Lambrigger, D. and W{\"u}thrich, M. (2009).
		Multivariate extremes and the aggregation of dependent risks: Examples and counter-examples.
		{\em Extremes}, \textbf{12}(2), 107--127.
		
		\bibitem[\protect\citeauthoryear{Embrechts et al.}{2018}]{ELW18}
		Embrechts, P., Liu, H. and Wang, R. (2018). Quantile-based risk sharing. \emph{Operations Research}, \textbf{66}(4), 936--949.
		
		
		
		
		
		\bibitem[\protect\citeauthoryear{Embrechts et al.}{Embrechts et al.}{2014}]{EPRWB14} {Embrechts, P., Puccetti, G., R\"uschendorf, L., Wang, R. and Beleraj, A.} (2014). An academic response to Basel 3.5. \emph{Risks}, \textbf{2}(1), 25--48.
		
		
		
		\bibitem[\protect\citeauthoryear{Embrechts et al.}{Embrechts et al.}{2015}]{EWW15} {Embrechts, P., Wang, B. and Wang, R.} (2015). Aggregation-robustness and model uncertainty of regulatory risk measures.  {\em Finance and Stochastics},  \textbf{19}(4), 763--790.
		
		
		\bibitem[\protect\citeauthoryear{Emmer et al.}{Emmer et al.}{2015}]{EKT15} {Emmer, S., Kratz, M. and Tasche, D.} (2015). What is the best risk measure in practice? A comparison of standard measures. \emph{Journal of Risk}, \textbf{18}(2), 31--60.
		
		
		
		
		
		
		
		
		
		
		\bibitem[\protect\citeauthoryear{F\"ollmer and Schied}{F\"ollmer and Schied}{2016}]{FS16} F\"ollmer, H.~and Schied, A.~(2016). \emph{Stochastic Finance. An Introduction in Discrete Time}. Fourth Edition.  {Walter de Gruyter, Berlin}.

		
		
		
		
		\bibitem[\protect\citeauthoryear{Han et al.}{Han et al.}{2022}]{HLW22} 
		Han, X., Lin, L. and Wang, R. (2022). Diversification quotients: Quantifying diversification via risk measures.   \emph{arXiv}: 2206.13679.
  
	{	\bibitem[\protect\citeauthoryear{Han et al.}{Hult and Lindskog}{2002}]{HL02} 		
	Hult, H. and Lindskog, F. (2002). Multivariate extremes, aggregation and dependence in elliptical distributions. \emph{Advances in Applied Probability}, \textbf{34}(3), 587--608.	}
		
		\bibitem[\protect\citeauthoryear{Gneiting}{Gneiting}{2011}]{G11}{ Gneiting, T.} (2011). Making and evaluating point forecasts. {\em Journal of the American Statistical Association},~{\bf 106}(494), 746--762.
		
		
		
		
		
		
		
		
		\bibitem[\protect\citeauthoryear{Ibragimov et al.}{2011}]{IJW11}
		Ibragimov, R., Jaffee, D. and Walden, J. (2011). Diversification disasters. {\it Journal of Financial Economics}, {\bf 99}(2), 333--348.
		
		
		
		
		
		
		
		
		
		

\bibitem[\protect\citeauthoryear{Koumou and Dionne}{2022}]{KD22}
Koumou, G. B. and  Dionne, G. (2022). Coherent diversification measures in portfolio theory: An axiomatic foundation. \emph{Risks}, \textbf{10}(11), 205.

{\bibitem[\protect\citeauthoryear{Kulik and Soulier}{2020}]{KS20}
  Kulik, R. and Soulier, P. (2020). \emph{Heavy-Tailed Time Series. }Berlin: Springer.}
		
		
		
		
		
		
		
		
		
		
{\bibitem[\protect\citeauthoryear{Lauzier et al.}{Lauzier et al.}{2023}]{LLW23}
  Lauzier, J. G., Lin, L. and  Wang, R. (2023). Pairwise counter-monotonicity. \emph{arXiv}: 2302.11701.}
		
{\bibitem[\protect\citeauthoryear{Lee and Ahn}{Lee and Ahn}{2014}]{LA14}
Lee, W. and Ahn, J.~Y. (2014).
On the multidimensional extension of countermonotonicity and its applications.
\emph{Insurance: Mathematics and Economics}, \textbf{56}, 68--79.}		
		

{\bibitem[\protect\citeauthoryear{Li and Wang}{Li and Wang}{2022}]{LW22}
Li, H. and Wang, R.  (2022). PELVE: Probability equivalent level of VaR and ES. \emph{Journal of Econometrics}, \textbf{234}(1), 353--370.}

{\bibitem[\protect\citeauthoryear{Li et al.}{Li et al.}{2018}]{LSWY18}
    Li, L., Shao, H., Wang, R. and Yang, J. (2018). Worst-case range value-at-risk with partial information. \emph{SIAM Journal on Financial Mathematics}, \textbf{9}(1), 190--218.}
		
		
		
		\bibitem[\protect\citeauthoryear{Liu et al.}{2021}]{LSW21}
		Liu, P., Schied, A. and Wang, R. (2021). Distributional transforms, probability distortions, and their applications. \textit{Mathematics of Operations Research}, \textbf{46}(4), 1490--1512.
		
		
		
		
		
		
		
		
		
		
		
		
		
		\bibitem[\protect\citeauthoryear{Mainik and Embrechts}{2013}]{ME13}
		Mainik, G. and Embrechts, P. (2013). Diversification in heavy-tailed portfolios: Properties and pitfalls. \emph{Annals of Actuarial Science}, \textbf{7}(1), 26--45.
		
		\bibitem[\protect\citeauthoryear{Mainik and R\"uschendorf}{2010}]{MR10}
		Mainik, G. and R\"uschendorf, L (2010). On optimal portfolio diversification with respect to extreme risks. \emph{Finance and Stochastics}, \textbf{14}(4), 593--623.
		
		
		
		\bibitem[\protect\citeauthoryear{Markowitz}{Markowitz}{1952}]{M52} Markowitz, H. (1952). Portfolio selection. \emph{Journal of Finance}, \textbf{7}(1), 77--91.
		
		
		
		
		
		
		
		
		\bibitem[\protect\citeauthoryear{McNeil et al.}{McNeil et al.}{2015}]{MFE15}
		{McNeil, A. J., Frey, R. and Embrechts, P.} (2015). \emph{Quantitative
			Risk Management: Concepts, Techniques and Tools}. Revised Edition.  Princeton, NJ:
		Princeton University Press.

  







 
		
		\bibitem[\protect\citeauthoryear{Puccetti and Wang}{Puccetti and 	Wang}{2015}]{PW15} Puccetti, G. and Wang R.  (2015). Extremal dependence concepts. \emph{Statistical Science},  \textbf{30}(4), 485--517.
		
		
		
		\bibitem[\protect\citeauthoryear{Rockafellar et al.}{Rockafellar et al.}{2014}]{RRM14}
		Rockafellar, R. T., Royset, J. O. and   Miranda, S. I. (2014). Superquantile regression with applications to buffered reliability, uncertainty quantification, and conditional value-at-risk. \emph{European Journal of Operational Research}, \textbf{234}(1), 140--154.
		
		

		
		
		
		
		
		
		\bibitem[\protect\citeauthoryear{R\"uschendorf}{R{\"u}schendorf}{2013}]{R13}		R{\"u}schendorf, L. (2013).		{\em Mathematical Risk Analysis. Dependence, Risk Bounds, Optimal			Allocations and Portfolios}.		Springer, Heidelberg.
		
		
		
		
		
		
		\bibitem[\protect\citeauthoryear{Tasche}{2007}]{T07}
		Tasche, D. (2007). Capital allocation to business units and sub-portfolios: The Euler principle. \emph{arXiv}: 0708.2542.
		
		
		
		
		
		
		\bibitem[\protect\citeauthoryear{Wang and Wang}{Wang and Wang}{2016}]{WW16}		Wang, B. and  Wang, R. (2016).		Joint mixability.  \emph{Mathematics of Operations Research}, \textbf{41}(3), 808--826.
		
		
		
		
		
		\bibitem[\protect\citeauthoryear{Wang and Zitikis}{2021}]{WZ20}
		Wang, R. and Zitikis, R. (2021). An axiomatic foundation for the Expected Shortfall. \emph{Management Science}, \textbf{67}(3), 1413--1429.

  \bibitem[\protect\citeauthoryear{Ziegel}{2016}]{Z16} Ziegel, J. F. (2016). Coherence and elicitability. \emph{Mathematical Finance}, \textbf{26}(4), 901--918.

  
		
		
	\end{thebibliography}
\end{document}